\documentclass[a4paper,11pt]{article}

\usepackage{authblk}
\usepackage{trackchanges}
\addeditor{IMM}
\usepackage{amsmath}
\usepackage{amssymb}
\usepackage{url}
\usepackage{hyperref}
\usepackage{color,graphicx}
\usepackage{algorithmic}
\usepackage[ruled,lined,algo2e]{algorithm2e}
\DeclareMathOperator*{\argmin}{arg\!\min}
\DeclareMathOperator*{\argmax}{arg\!\max}
\usepackage{multirow}
\usepackage{mathtools}
\usepackage{lineno}
\usepackage{comment}

\usepackage{mwe}
\usepackage{subfig}
\usepackage{adjustbox}
\usepackage{float}

\newtheorem{theorem}{Theorem}
\newtheorem{proposition}[theorem]{Proposition}
\newtheorem{observation}[theorem]{Observation}
\newtheorem{lemma}[theorem]{Lemma}

\newcommand{\qed}{\hspace*{\fill} $\square$  \ifmmode \else
    \par\addvspace\topsep\fi}
\newenvironment {proof}{\par\addvspace\topsep\noindent{\it Proof.}
    \ignorespaces }{\qed}
\newcommand{\claimqed}{\hspace*{\fill} $\triangle$  \ifmmode \else
    \par\addvspace\topsep\fi}

\providecommand{\keywords}[1]{\textbf{\textit{Keywords:}} #1}

\begin{document}


\title{Spectrum graph coloring and applications to WiFi channel assignment}







\author[1]{David Orden
\thanks{Email: david.orden@uah.es.}}

\author[2]{Jose Manuel Gimenez-Guzman
\thanks{Email: josem.gimenez@uah.es.}}

\author[2]{Ivan Marsa-Maestre
\thanks{Email: ivan.marsa@uah.es.}}

\author[2]{Enrique de la Hoz
\thanks{Email: enrique.delahoz@uah.es.}}

\affil[1]{Departamento de F\'{\i}sica y Matem\'aticas, Universidad de Alcal\'a, Spain, \texttt{david.orden@uah.es}}
\affil[2]{Departamento de Autom\'atica, Universidad de Alcal\'a, Spain, \texttt{[josem.gimenez|ivan.marsa|enrique.delahoz]@uah.es}}

\date{}

\maketitle


\begin{abstract}
We introduce and explore a family of vertex-coloring problems which, surprisingly enough, have not been considered before despite stemming from the problem of Wi-Fi channel assignment. Given a spectrum of colors, endowed with a matrix of interferences between each pair of colors, the \textsc{Threshold Spectrum Coloring} problem fixes the number of colors available and aims to minimize the interference threshold, i.e., the maximum of the interferences at the vertices. Conversely, the \textsc{Chromatic Spectrum Coloring} problem fixes a threshold and aims to minimize the number of colors for which respecting that threshold is possible. As main theoretical results, we prove tight upper bounds for the solutions to each problem.
Since both problems turn out to be NP-hard, we complete the scene with experimental results. We propose a DSATUR-based heuristic and study its performance to minimize the maximum vertex interference in Wi-Fi channel assignment, both for randomly generated graphs and for a real-world scenario. Further, for all these graphs we experimentally check the goodness of the theoretical bounds.
\end{abstract}



 	

\section{Introduction}\label{sec:spectrum}

Graph coloring is, undoubtedly, one of the main problems in Discrete Mathematics, attracting researchers from both mathematics and engineering because of its theoretical challenges and its applications~\cite{mt-surveyVCP-10, tuza-inHandbook-03}. One of the most prominent applications of vertex-coloring problems is frequency assignment~\cite{ahkms-FAP-07}, with a huge variety of models ranging from the most naive, forbidding monochromatic edges,
to the most general, with assignment constraints, interference constraints, and an objective function.

The present work introduces and explores an intermediate model, which is shown to be interesting enough for mathematicians and useful enough for engineers: An abstract graph~$G$ and a spectrum of colors~$S=\{c_1,\ldots,c_{s}\}$ endowed with a matrix~$W$
of non-negative distances $W_{ij}=W(c_i,c_j)$ between each pair of colors, playing the role of interferences. Throughout this paper, the graph~$G$ will be undirected and, thus, the matrix~$W$ will be symmetric. For such a pair $(G,W)$, a coloring $c$ of the graph will induce at each vertex~$v$ an \emph{interference}
\[I_v(G,W,c)=\sum_{u\in N(v)}W(c(u),c(v)).\]

This model stems naturally from the problem of Wi-Fi channel assignment.
We all use the IEEE 802.11 (Wi-Fi) technology, which has been widely deployed in local area networks, mostly because of both its low cost and the use of unlicensed frequency bands.
However, for a successful deployment there are some performance issues to be managed, such as avoiding an excessive interference that will impact performance.
There is a direct relation between interference and perceived throughput~\cite{Bazzi11}, so that minimizing the interference would result in maximizing the available throughput.
This is why a frequency channel has to be selected for each access point (AP),
accounting for interferences between channels and trying to minimize the total interference in order to maximize performance.
Similar goals have been considered recently in the literature for related problems, like bandwidth allocation in cellular networks~\cite{bar2015bandwidth}.

In the recent past we have successfully used our model to find efficient frequency assignments both in Wireless
Surveillance Sensor Networks~\cite{hgmo-WSSNs-15} and in real-world scenarios~\cite{AAMAS-WiFi}. We have even applied our model to reconfigure a critical network in the event of a security incident~\cite{AAMAS-Resilience}. With the present paper we aim to formally introduce a mathematically rigorous description of the model, exploring two novel coloring problems for which we provide theoretical bounds, show their complexity, and finally propose and test a heuristic.

In particular, the remainder of this section introduces the two novel coloring problems and presents a case study to illustrate both of them, finishing with a compilation of relevant related work. Section~2 is devoted to theoretical results, proving upper bounds for the solutions to each problem. Section~3 starts showing that both problems are NP-hard, so a DSATUR-based heuristic is proposed for each problem and experimental results are performed to compare these heuristics with selected reference approaches for randomly generated graphs. Section~4 goes a step further by showing experimental results for a realistic setting. Finally, Section~5 summarizes our conclusions.

\subsection{Two novel coloring problems}
\label{subsec:TSC-CSC}

First, we introduce the problem which best fits the setting of Wi-Fi channel assignment, which we call \textsc{Threshold Spectrum Coloring} (TSC) problem: Given a graph~$G$ and a spectrum of~$k$ colors (channels), endowed with a $k\times k$ matrix~$W$ of interferences between them, the goal is to determine the minimum threshold~$t\in\mathbb{R}_{\geq 0}$ such that $(G,W)$ admits a $k$-coloring~$c$ (assignment of channels) in which the interference at every vertex is at most~$t$, i.e., $I_v(G,W,c)\leq t,\ \forall v$. Such a minimum~$t$ will be called the \emph{minimum $k$-chromatic threshold} of~$(G,W)$, denoted as~$T_k(G,W)$.

Then, since the TSC problem fixes the parameter~$k$ and aims to minimize the parameter~$t$, it is natural to wonder about the complementary problem: In the \textsc{Chromatic Spectrum Coloring} (CSC) problem, a threshold $t\in\mathbb{R}_{\geq 0}$ is fixed and the spectrum is let to have as size the number $|V(G)|$ of vertices, the goal being to determine the minimum number of colors (channels)~$k\in\mathbb{N}$ such that $(G,W)$ admits a $k$-coloring~$c$ (assignment of channels) in which the interference at every vertex is at most that threshold~$t$. Such a minimum~$k$ will be called the \emph{$t$-interference chromatic number} of~$(G,W)$, denoted as~$\chi_t(G,W)$.
This CSC problem has also a translation in the context of frequency assignment: It aims to compute the minimum number of frequencies (colors) that guarantee a minimum throughput to every user in the network, so it can be seen as the problem of determining the number of resources required to provide a certain quality of service.

\subsection{A case study}
\label{subsec:CaseStudy}

Let us illustrate these two problems by analyzing an easy, but non-trivial, example using the paw graph~$PG$~\cite{west-introGraphTheory-01}, see Figure~1 (right).

\begin{figure}[!htb]
\begin{center}
\includegraphics[width=\textwidth]{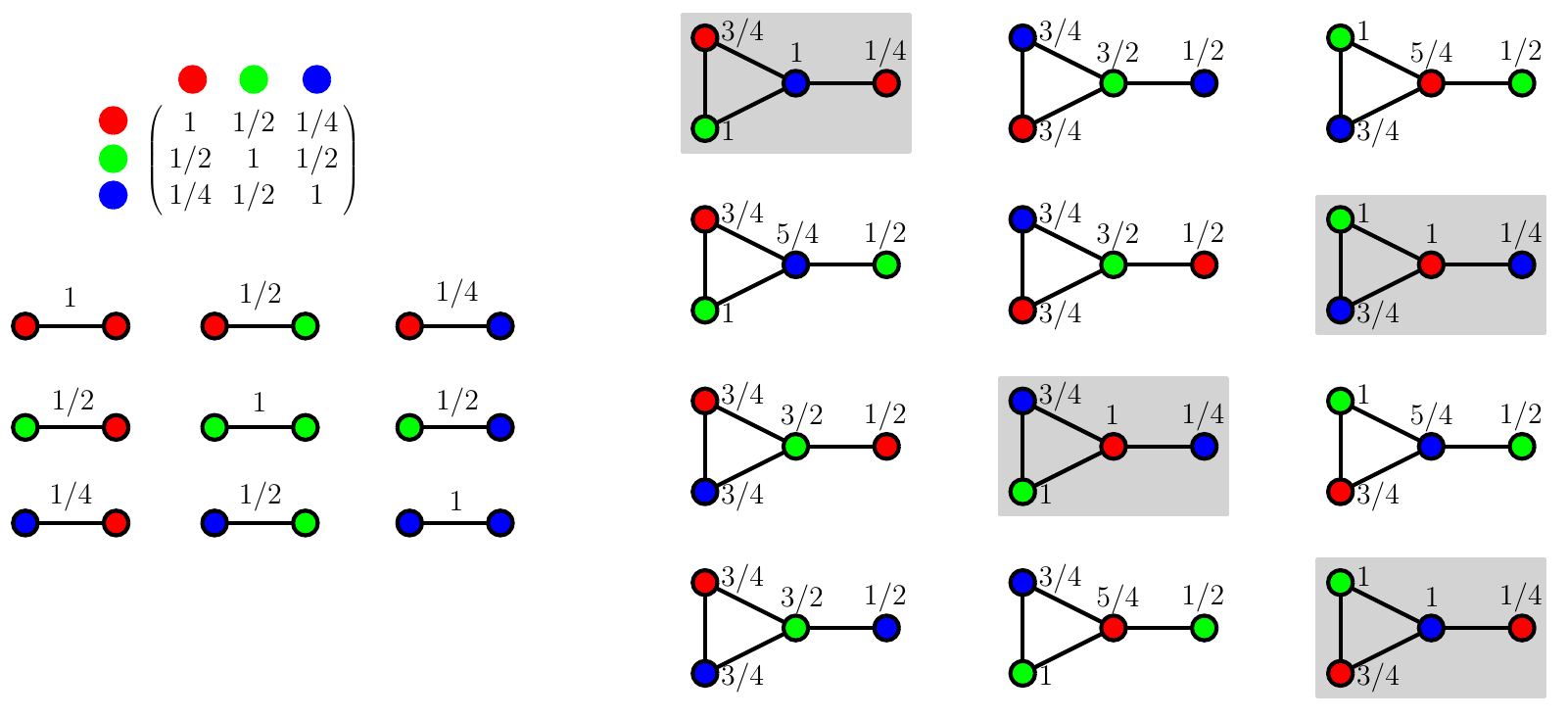}
\end{center}
\caption{Left: Matrix $W_{2ed}$ and the interferences it induces on the different possible colorings of an edge. Right: All proper 3-colorings of the paw graph and the interferences they induce at every vertex. Highlighted colorings achieve interference~$\leq 1$ at every vertex.} \label{fig:TotalInterferenceColoring}
\end{figure}

For the TSC problem, consider an spectrum of $k=3$ colors $$S=\{\mbox{red},\mbox{green},\mbox{blue}\}$$ endowed with the $3\times 3$ matrix of interferences with exponential decay of base~$2$,
\[
W_{2ed}=
\left(
\!\begin{array}{ccc}1 & 1/2 & 1/4\\1/2 & 1 & 1/2\\1/4 & 1/2 & 1
\end{array}
\!\right).
\]
See Figure~\ref{fig:TotalInterferenceColoring} (left) for an illustration.
In this TSC problem, we have a fixed number of colors~$k$ and we want to find the smallest possible interference threshold~$t$. Figure~\ref{fig:TotalInterferenceColoring}, right, shows all the proper (without monochromatic edges) $3$-colorings of the paw graph, together with the interference $I_v(PG,W_{2ed},c)$ at every vertex~$v$.
Observe that the highlighted colorings achieve interference at most~$t=1$ at every vertex. This is impossible for improper colorings, since the matrix~$W$ assigns interference~$1$ to monochromatic edges which, therefore, will lead to one of its endpoints having interference greater than~$1$.

Let us show now that, in this setting, no coloring can achieve a maximum vertex interference~$t$ strictly smaller than~$1$: Focus on the central vertex of the paw graph and observe that, according to the matrix~$W_{2ed}$, that vertex can only achieve an interference smaller than~$1$ if it is colored red and its neighbors are colored blue, or vice versa. But this gives rise to a monochromatic edge between the left vertices of the paw, which implies interference greater than~$1$ at those two vertices. See Figure~\ref{fig:ExampleTSC}.

\begin{figure}[!htb]
\begin{center}
\includegraphics[scale=0.5]{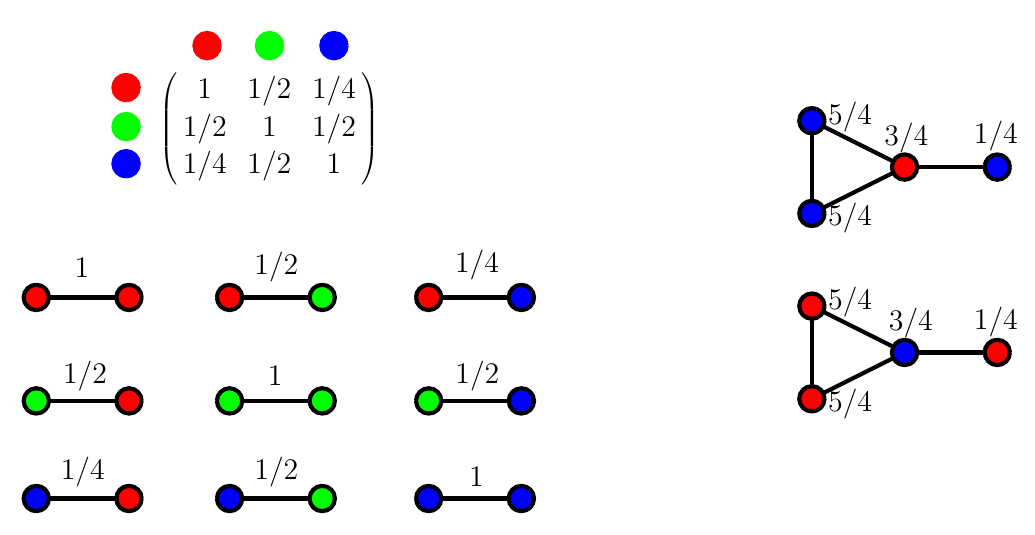}
\end{center}
\caption{Possibilities for the central vertex of the paw graph to have interference smaller than~$1$.} \label{fig:ExampleTSC}
\end{figure}

Therefore, we have solved the corresponding TSC problem, proving that for the paw graph~$PG$ with $k=3$ colors and a matrix~$W_{2ed}$ with exponential decay of base~$2$, the minimum chromatic threshold~$t$ achievable is~$1$, i.e., \[T_3(PG,W_{2ed})=1.\]

For the CSC problem, the size of the spectrum is let to equal the number of vertices, which is four in the paw graph~$PG$. Consider for example~$S=\{\mbox{red},\mbox{green},\mbox{blue},\mbox{violet}\}$, also endowed with the  matrix of interferences with exponential decay of base~$2$, in this case of size $4\times 4$,

\[
W_{2ed}=
\left(
\!\begin{array}{cccc}
1 & 1/2 & 1/4 & 1/8\\
1/2 & 1 & 1/2 & 1/4\\
1/4 & 1/2 & 1 & 1/2\\
1/8 & 1/4 & 1/2 & 1\\
\end{array}
\!\right).
\]

In this setting, one can fix an interference threshold~$t=1$ and aim to find the smallest number~$k$ of colors for which such a threshold is achievable. The discussion above for the previous problem shows that the threshold~$t=1$ can be achieved with three of those colors, i.e., $k=3$. We now show that such a threshold is impossible to achieve with only $k=2$ colors among the four colors available: Focus again on the central vertex of the paw graph~$PG$. If it has a neighbor with its same color, then a monochromatic edge arises and, as above, this implies the interference at some vertex being greater than one. If the central vertex has no neighbor with its same color, then all its neighbors share a common color, so that again a monochromatic edge arises and there is a vertex with interference greater than one. See Figure~\ref{fig:ExampleCSC}.

\begin{figure}[htb]
\begin{center}
\includegraphics[width=\textwidth]{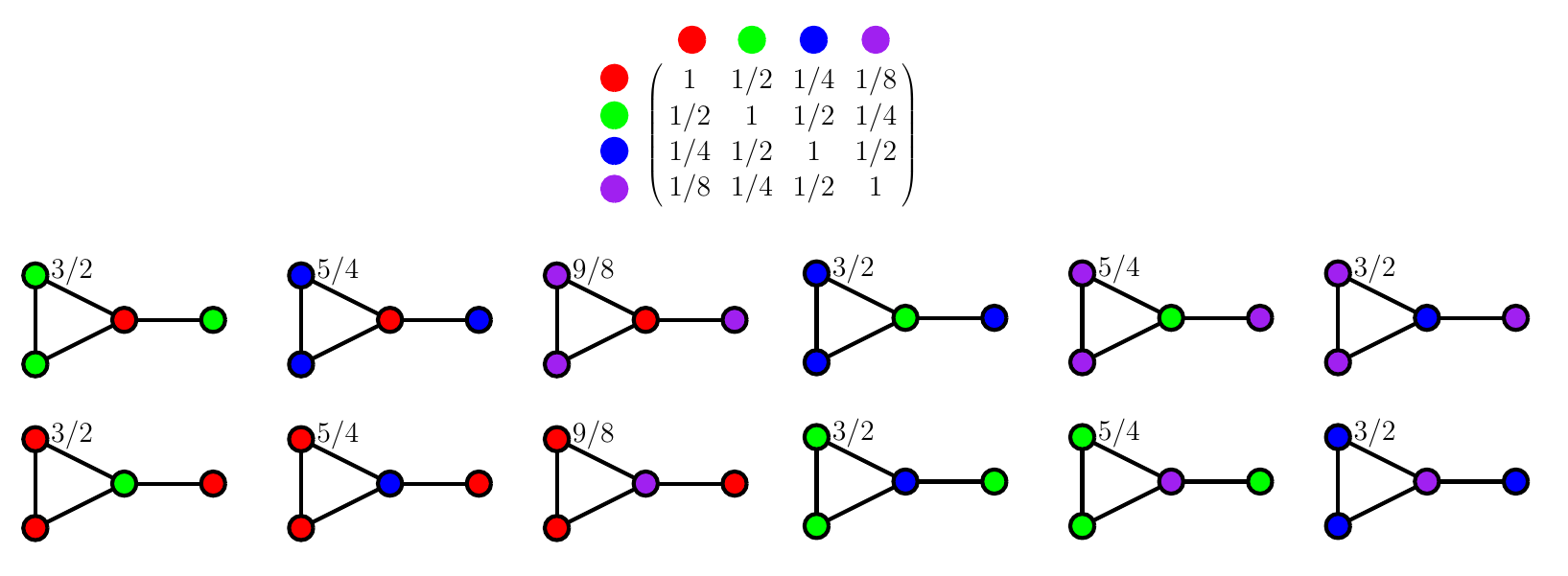}
\end{center}
\caption{Possibilities for the central vertex of the paw graph to have no neighbor with its same color, when using~$k=2$ colors among the~$4$ colors available.} \label{fig:ExampleCSC}
\end{figure}

Therefore, we have solved the corresponding CSC problem proving that, for the paw graph~$PG$ with threshold $t=1$ and a matrix~$W_{2ed}$ with exponential decay of base~$2$, the interference chromatic number is~$3$, i.e., \[\chi_1(PG,W_{2ed})=3.\]

\subsection{Related work}

Araujo et al.~\cite{abghmm-WeightedImproper-12} consider a weight function~$w$ on the edges of the graph~$G$ instead of a matrix~$W$ of interferences between colors, defining the interference at a vertex to be the sum of weights of incident monochromatic edges. Although their bounds are similar to ours, using analogous techniques, we further show tightness and follow a different, more detailed, scheme.
Many other works impose conditions to the colors of the endpoints of any edge. Most of them can be framed into \emph{$L(p_1,\ldots,p_k)$-labellings}~\cite{gk-GraphLabellings-09}, where vertices at distance~$i$ must get colors at distance at least~$p_i$. Particular instances include \emph{distance coloring}~\cite{sharp-distancecoloring-07}, where no two vertices at distance at most~$d$ can have the same color, \emph{$\lambda$-coloring}~\cite{bktl-lambdacoloring-00}, where adjacent vertices must get colors at least~$2$ apart and vertices at distance at most~$ 2$ must get different colors, as well as \emph{$L(h,k)$-labellings}~\cite{c-LhkLabellingSurvey-11}, where adjacent vertices must get colors at least~$h$ apart and vertices at distance~$2$ must get colors at least~$k$ apart.

For another setting, in a \emph{bandwidth coloring}~\cite{mt-surveyVCP-10} a distance~$d_{ij}$ is associated to each particular edge~$ij$, forcing its endpoints to get colors at least~$d_{ij}$ apart.
A similar flavor has the \emph{$T$-coloring}~\cite{roberts-Tcolorings-91}, where adjacent vertices must get colors whose distance is not in a prescribed set~$T$, and the \emph{$S$-coloring}~\cite{gastineau-Scoloring-15}, where the vertices have to be partitioned into sets with prescribed pairwise distances.

Of a different nature is the approach of coloring the edges of the graph instead of the vertices~\cite{edge-coloring1986}, an approach which has been applied to minimize the number of distinct channels used to route messages in a network~\cite{rainbow2017}.

Finally, of independent interest are the works on the Frequency Assignment Problem, which use different models to fit the characteristics of specific applications. The bibliography about this problem is too extensive to include here even a selection. Instead, we prefer referring the reader to the comprehensive survey by Aardal et al.~\cite{ahkms-FAP-07}.

\section{Theoretical results}

In this section we prove theoretical upper bounds for the goal of the TSC problem, the minimum $k$-chromatic threshold~$T_k(G,W)$, and for the aim of the CSC problem, the $t$-interference chromatic number~$\chi_t(G,W)$.

Given a graph~$G$, a matrix~$W$, and a coloring~$c$, we denote the \emph{potential interference} at vertex~$v$ if this vertex had color~$i$ as \[I_v^i(G,W,c)=\sum_{u\in N(v)}W(c(u),i).\] Note that for~$i$ the actual color of~$v$, i.e., $i=c(v)$, we get the  actual interference at vertex~$v$, that is, $I_v^{c(v)}(G,W,c)=I_v(G,W,c)$.

We will say that a $k$-coloring~$c$ of~$G$ is \emph{$W$-stable} if, for every vertex, the actual interference is not greater than any of the potential interferences, i.e., if
for every vertex~$v$ we have $I_v(G,W,c)\leq I_v^{j}(G,W,c)$ for all $j\in\{1,\ldots,k\}$. We first prove that stable colorings do exist:

\begin{proposition}
\label{proposition:stableColoring}
Given a graph~$G$ and a spectrum~$S$ of size $|S|\geq 2$ endowed with a matrix~$W$ of interferences, there exists a $W$-stable coloring of~$G$.
\end{proposition}
\begin{proof}
Let~$I$ be the sum of edge-interferences for the current coloring~$c$, i.e., \[I=\sum_{uv\in E}W(c(u),c(v)).\]
(Note that this is actually half of the sum of interferences at every vertex.) Observe that any coloring achieving the minimum for this sum of edge-interferences will be a $W$-stable coloring of~$G$: If there is a current interference~$I_v^{c(v)}$ which is greater than a potential interference~$I_v^{i}$, then recoloring would lead to a smaller sum~$I$ of edge-interferences.
Furthermore, we can prove that such a stable coloring can be found by a random greedy algorithm: Start with a random coloring of the graph~$G=(V,E)$ and, if there exist a vertex $v\in V$ and a color~$j\in S$ such that the potential interference~$I_v^{j}(G,W,c)$ for that color is smaller than the current interference~$I_v^{c(v)}(G,W,c)$, then recolor the vertex~$v$ with color~$j$ and repeat.

Each step of the procedure increases~$I$ by $I_v^{j}(G,W,c)$ and decreases~$I$ by $I_v^{c(v)}(G,W,c)$ (in both cases, $c$ is the initial coloring before the step). Hence, the sum of edge-interferences~$I$ decreases by a positive amount $I_v^{c(v)}(G,W,c)-I_v^{j}(G,W,c)$ and the random greedy algorithm finishes with a $W$-stable coloring of~$G$.
\end{proof}

Before proving our bounds, we also need the following lemma:
\begin{lemma}
\label{lemma:stable}
Any $W$-stable $k$-coloring~$c$ of a graph~$G$ fulfills that, for each vertex~$v$,
\[
k\ \!I_v^{c(v)}(G,W,c)\leq
\operatorname{deg}(v)\ \!||W||_{\infty},
\]
where $\displaystyle ||W||_{\infty}=\max_{i}\sum_{j} W_{ij}$ is the natural norm of the matrix~$W$.
\end{lemma}
\begin{proof}
First, observe that such a coloring fulfills
\[
k\ \!I_v^{c(v)}(G,W,c)\leq
\sum_{j=1}^k I_v^{j}(G,W,c),
\]
because $c$ being stable implies that the interference $I_v^{c(v)}(G,W,c)$ is actually not greater than all of the potential interferences $I_v^{j}(G,W,c)$. Restating the sum of the potential interferences at~$v$ as the sum of interferences contributed by neighbors of each color while varying the color of~$v$, the right-hand side in the previous equation equals
\[
\sum_{u\in N(v) | c(u)=1} \sum_{j} W_{1j}\ +\ \cdots\ +\ \sum_{u\in N(v) | c(u)=k} \sum_{j} W_{kj}
\]
and this sum is not greater than the number of neighbors of~$v$ times the maximum of the sums $\sum_{j} W_{ij}$, i.e., not greater than $\operatorname{deg}(v)\ \!||W||_{\infty}$.
\end{proof}

Now we are ready to prove an upper bound for the aim of the TSC problem:

\begin{theorem}
\label{theorem:TSCbound}
Given a graph~$G$ and a spectrum~$S$ of size $|S|$ endowed with a matrix~$W$ of interferences, for any fixed natural number~$2\leq k \leq |S|$,
the following bound holds for the minimum $k$-chromatic threshold~$T_k(G,W)$:
\[
T_k(G,W)\leq \frac{\Delta(G)\ ||W||_{\infty}}{k},
\]
where $\Delta(G)$ is the maximum vertex-degree in the graph~$G$.
Furthermore, this bound is tight.
\end{theorem}
\begin{proof}
By definition of~$T_k(G,W)$, in order to prove the bound it is enough to prove that there exists a coloring~$c_0$ of~$G$ with~$k$ colors for which the interference at every vertex does not exceed the threshold~$t_0=\frac{\Delta(G)\ ||W||_{\infty}}{k}$.
Because of Proposition~\ref{proposition:stableColoring} and $k\geq 2$, we know that there exists a $W$-stable coloring~$c_0$ of~$G$ using~$k$ colors. It just remains to prove the interference condition:
For every vertex~$v$ with maximum interference we have
\[
k\ \! T_k(G,W)\leq
k\ \!I_v^{c_0(v)}(G,W,c_0) \leq
\operatorname{deg}(v)\ \!||W||_{\infty},
\]
where the leftmost inequality follows from $T_k(G,W)$ being minimum and the rightmost inequality comes from Lemma~\ref{lemma:stable}. The bound in the statement follows.

For the final claim in the statement, it is enough to note that taking $k=2$ and $W=I$ our bound results in $T_2(G,W)\leq \left\lceil\frac{\Delta(G)}{2}\right\rceil$, which is tight for the graph~$G$ being a cycle of odd length. In that case it is not possible to avoid monochromatic edges, with the best coloring achieving a single monochromatic edge and, therefore, a smallest threshold of~$1$, which matches the bound $\left\lceil\frac{\Delta(G)}{2}\right\rceil=\left\lceil\frac{2}{2}\right\rceil=1$.
\end{proof}

Next, we provide an upper bound for the aim of the CSC problem. The bound uses the generalization of the greatest common divisor (gcd) to non-integer numbers, which is defined analogously just considering that a number~$x$ divides a number~$y$ if the fraction~$y/x$ is an integer.

\begin{theorem}
\label{theorem:CSCbound}
Given a graph~$G$ and a spectrum~$S$ having size $|S|\geq 2$ and endowed with a matrix~$W$ of interferences, for any fixed threshold~$t$ being a multiple of~$\gcd(W)$ and such that
$|S| t \geq \Delta(G)\ ||W||_{\infty}$
the following bound holds for the $t$-interference chromatic number~$\chi_t(G,W)$:
\[
\chi_t(G,W)\leq \left\lceil\frac{\Delta(G)\ ||W||_{\infty}+\gcd(W)}{t+\gcd(W)}\right\rceil.
\]
If $t$ is not a multiple of~$\gcd(W)$, the bound can be expressed as
\[
\chi_t(G,W)\leq \left\lceil\frac{\Delta(G)\ ||W||_{\infty}+\gcd(W)}{\gcd(W)\left\lfloor\frac{t}{\gcd(W)}\right\rfloor+\gcd(W)}\right\rceil,
\]
where~$t$ is replaced by the nearest multiple of~$\gcd(W)$ below~$t$.
Furthermore, these bounds are tight.
\end{theorem}
\begin{proof}
The number of colors has to be at least one and the bound is trivially satisfied for $\chi_t(G,W)=1$, so let us focus on the case $\chi_t(G,W)>1$. We start with the case of~$t$ being a multiple of~$\gcd(W)$.

By definition of~$\chi_t(G,W)$, in order to prove the bound it is enough to prove that there exists a coloring~$c_0$ of~$G$ with $k_0=\left\lceil\frac{\Delta(G)\ ||W||_{\infty}+\gcd(W)}{t+\gcd(W)}\right\rceil$ colors for which the interference at every vertex does not exceed the threshold~$t$, i.e., $I_v(G,W,c_0)\leq t,\ \forall v$. Note that the condition on the threshold in the statement ensures that $k_0\leq |S|$.

Because of Proposition~\ref{proposition:stableColoring} and $k_0\geq 2$, we know that there exists a $W$-stable coloring~$c_0$ of~$G$ using~$k_0$ colors. Hence, it just remains to prove the interference condition.

By contradiction, suppose that for the coloring~$c_0$ there is a vertex~$v$ in~$G$ with an interference above the threshold~$t$, i.e., \[
I_v^{c_0(v)}(G,W,c_0)> t
\]
which, because of~$t$ and the interferences around a vertex being multiples of~$\gcd(W)$, implies
\[
I_v^{c_0(v)}(G,W,c_0) \geq t+\gcd(W).
\]
Then, since the coloring~$c_0$ is stable, Lemma~\ref{lemma:stable} leads to
\[
\Delta(G)\ \!||W||_{\infty} \geq
\operatorname{deg}(v)\ \!||W||_{\infty} \geq
k_0\ \!I_v^{c_0(v)}(G,W,c_0) \geq
k_0\ \!(t+\gcd(W)),
\]
where the leftmost inequality follows from the definition of~$\Delta(G)$.

But then our choice of~$k_0$ implies that
\[
\Delta(G)\ \!||W||_{\infty} \geq
\Delta(G)\ ||W||_{\infty}+\gcd(W),
\]
which is a contradiction.

The same arguments work for the case of~$t$ not being a multiple of~$\gcd(W)$, taking into account that interferences around a vertex have to be multiples of~$\gcd(W)$. For the final claim in the statement, it is enough to take $W=I$ and $t=0$ (which is possible as long as $|S|\geq \Delta(G)+1$, i.e., the coloring is not forced to be improper). Thus, our bound coincides with Brooks' bound~\cite{brooks-maxdegree-41} for the chromatic number $\chi_0(G,I)=\chi(G)\leq \Delta(G)+1$, which is tight for the graph~$G$ being a cycle of odd length.
\end{proof}

The reader can check that the case analyzed in Subsection~\ref{subsec:CaseStudy} does fulfill the bounds in Theorems~\ref{theorem:TSCbound} and~\ref{theorem:CSCbound}. Respectively:
\[
1=T_3(PG,W_{2ed})\leq\frac{\Delta(PG)\ ||W_{2ed}||_{\infty}}{3}=
\frac{3\cdot 2}{3}=2
\]
and
\[
3=\chi_1(PG,W_{2ed})\leq
\left\lceil\frac{\Delta(PG)\||W_{2ed}||_{\infty}+\gcd(W_{2ed})}{1+\gcd(W_{2ed})}\right\rceil=
\left\lceil\frac{3\cdot \frac{9}{4}+\frac{1}{8}}{1+\frac{1}{8}}\right\rceil=7.
\]
Note that, in this case, the small size of the graph makes the upper bound~$7$ exceed the actual size of the spectrum~$|S|=|V(G)|=4$, because the chosen threshold $t=1$ does not fulfill the condition in the statement of Theorem~\ref{theorem:CSCbound}:
\[
4 = 4\cdot 1 = |S|\cdot t <
\Delta(PG)\cdot ||W_{2ed}||_{\infty}
=
3\cdot \frac{9}{4} =
\frac{27}{4} = 6.75.
\]

\begin{observation}
\label{obs:ExtendSpectrum}
The condition on the threshold~$t$ in the statement of Theorem~\ref{theorem:CSCbound} reflects the fact that a trade-off is needed between the threshold~$t$ and the size~$|S|$ of the spectrum, since decreasing one of them might need an increase in the other. Furthermore, the theorem still has useful practical implications for values of the threshold~$t$ not fulfilling that condition: If the user was allowed to extend the spectrum~$S$ in such a way that (i) the bound given in the theorem is now below the size~$|S|$ and (ii) neither~$||W||_{\infty}$ nor~$\gcd(W)$ are increased, then the threshold~$t$ would fulfill the condition. Thus, the user would get not only an upper bound for the $t$-interference chromatic number, but also a certificate that for such a~$t$ the problem would be solvable for the extended spectrum. Finally, let us mention that the condition can be further refined to $|S| t \geq \Delta(G)\ ||W||_{\infty}-\gcd(W)(|S|-1)$
\end{observation}

It is interesting to note that, as long as the size of the spectrum is large enough for the chosen threshold, our bounds do not depend on the number of vertices in the graph~$G$, but only on its maximum degree~$\Delta(G)$. For the sake of applications to Wi-Fi channel assignment, this guarantees the scalability of the bounds for a growing number of APs as long as the maximum number of interfering APs can be restrained. Such a property has interesting implications in the design and planning of wireless networks infrastructures, since the transmitting power of the APs could be adjusted to guarantee that the degree of the graph stays below the necessary value, and therefore, to satisfy the maximum interference threshold for a given performance requirement.

\section{Experimental results for random graphs}
\label{results}

The well-known \textsc{Vertex Coloring} (VC) problem, which aims for a proper coloring (with no monochromatic edges), is a particular case of our \textsc{Chromatic Spectrum Coloring} problem, using the identity~$I$ as matrix of weights and a threshold $t=0$. Therefore, the usual chromatic number is the $0$-interference chromatic number of~$(G,I)$, that is, \[\chi(G)=\chi_0(G,I).\]
With the same identity matrix and a threshold~$t$, we get the \textsc{Improper Coloring} problem, which allows at most~$t$ monochromatic edges around a vertex. These two particular cases are NP-hard~\cite{karp-NPhard-72,woodall-improper-90}, hence so is our \textsc{Chromatic Spectrum Coloring} problem.
In addition, fixed any $k\geq 2$ it is NP-complete to decide if there exists a $k$-coloring with interference at most any threshold~$t\geq 2$, see~\cite{abghmm-WeightedImproper-12} and the references therein. This implies that our \textsc{Threshold Spectrum Coloring} is also NP-hard.

The hardness of the VC problem has given rise to many heuristics and metaheuristics for different vertex-coloring problems, e.g.,~\cite{bk-rchgc-2003,ehh-gcaea-98,s-newDSATUR-12}. With the TSC and CSC problems introduced in this paper being even more general, this section is devoted to experimentally test several techniques, in a range of scenarios, and to compare their performance for the TSC and CSC problems.

As a bonus, these experiments will allow to check how tight or loose are, for average scenarios, the theoretical upper bounds given in Theorems~\ref{theorem:TSCbound} and~\ref{theorem:CSCbound}.

\subsection{Settings}

Although we are aware of the number of benchmarks available like, e.g., those in DIMACS~\cite{DIMACS1996}, the fact that this is a seminal work, introducing a new type of coloring from both theoretical and applied perspectives, has led us to start the experiments by considering the widely used Erd\H{o}s-Renyi (ER) random graphs~\cite{newman02}. In these graphs, every pair of vertices
has a prescribed probability~$p$ of being connected. We have created different categories of graphs, by varying the number of vertices $n$ and the probability of connection $p$. In particular, we generated 10 graphs for each combination of the same settings used in \cite{mns-DSATURbased-15,s-newDSATUR-12}, that is, number of vertices $n\in\{60,70,80\}$ and  probability of connection $p \in \{0.1, 0.3, 0.5, 0.7, 0.9\}$, performing 20 repetitions of each experiment per graph.

For the matrix of interferences between colors, we have considered the exponential decay of base~$2$ already used in Subsection~\ref{subsec:CaseStudy}, $W_{ij}={1}/{2^{|i-j|}}$. Other choices have also been explored, leading to similar results which are not included here in order to avoid an excessive length of the paper.

For the TSC problem we have considered $k\in\{4,6,11\}$ as values for the number of colors in the spectrum, while for the CSC problem we have considered thresholds $t\in\{np/4,np/2,3np/4\}$, being~$np$ the expected average degree of the graph. These choices have a twofold interest: On one hand, because of including the number $k=11$ of channels available in the Wi-Fi problem, so that our results show interference thresholds that can be achieved at every vertex. On the other hand, because the choices for~$k$ turn out to be similar to the best number of colors obtained for the chosen values of~$t$ and vice versa.

\subsection{Coloring techniques used for benchmarking\label{sec:techniques}}

For an experimental analysis of our two problems, we propose heuristics inspired by the sequential greedy algorithm DSATUR~\cite{b-DSATUR-79} used for the \textsc{Vertex Coloring} problem.
Our implementation {TSC-DSATUR}
for the \textsc{Threshold Spectrum Coloring} problem looks for the color minimizing the interference at~$v$, as shown in Algorithm \ref{alg:TSC-DSATUR}. The algorithm starts with an undefined coloring~(1) and iterates by selecting at each iteration the uncolored vertex with the highest saturation degree, that is, the one with more already-colored neighbors, the one with highest degree in the case of a tie, or a random vertex among the highest-degree vertices in the case of a double tie~(2). Once a vertex $v$ has been selected, a color is assigned from the available set so that the interference at vertex~$v$ is minimized~(3).

\begin{algorithm2e} \caption{TSC-DSATUR coloring algorithm} \label{alg:TSC-DSATUR} 
\KwIn{\\
\Indp $G=(V,E)$: graph to be colored;
$S=\{c_i\}$: spectrum of colors\\
$W$: matrix of interferences\\
$k \mid 2\leq k\leq |S|$: maximum color number to be used from the spectrum\\
}
\KwOut{\\
\Indp $c$: $k$-coloring of the graph $G$}
$c(v)\coloneqq\varnothing,\ \forall v \in V$; \\
\lnl{InRes1}\While{$\exists v \in V \mid c(v)=\varnothing$ }{
\lnl{InRes2}$v=\argmax_{x \in V; c(x)=\varnothing} saturation\_degree(x)$;\\
\lnl{InRes3}$c(v)\coloneqq \argmin_{c_i \mid i\leq k}{\sum_{u\in N(v);c(u)\neq\varnothing}W(c(u),c_i)}$ \\
}
\end{algorithm2e}

Our implementation {CSC-DSATUR}
for the \textsc{Chromatic Spectrum Coloring} problem is shown in Algorithm \ref{alg:CSC-DSATUR}. The algorithm looks for a color that does not make interference at~$v$ exceed the product of the fixed threshold~$t$ by the proportion of neighbors already colored~(1). That is, if vertex~$v$ has five neighbours, but only three of them have been colored so far, the color chosen should guarantee the interference at~$v$ to be at most~$\frac{3}{5}t$. For a color to be chosen, this constraint must hold not only for~$v$, but also for any neighbor of~$v$ which has already been colored~(2). This is a conservative approach, to ensure that it is not possible to have a solution where the maximum interference per vertex is above the threshold.

\begin{algorithm2e} \caption{CSC-DSATUR coloring algorithm} \label{alg:CSC-DSATUR} 
\KwIn{\\
\Indp $G=(V,E)$: graph to be colored;
$S=\{c_i\}$: spectrum of colors\\
$W$: matrix of interferences\\
$t$: threshold on the maximum interference per vertex\\
}

\KwOut{\\
\Indp $c$: coloring of the graph $G$, $c(v)\coloneqq\varnothing,\ \forall v \in V$ if no valid coloring is found }
$c(v)\coloneqq\varnothing,\ \forall v \in V$; \\
\While{$\exists v \in V \mid c(v)=\varnothing$ }{
$v=\argmax_{x \in V; c(x)=\varnothing} saturation\_degree(x)$\\
$i=1;I_{max}\coloneqq \infty$\\
\lnl{InRes1}\While{$I_{max}>
\frac{|\{v\in N(v) \mid c(v)\neq\varnothing\}|}{|N(v)|}t;i\leq |S|$}{
$c(v)\coloneqq c_i$\\
$I_{max}=\sum_{u\in N(v);c(u)\neq\varnothing}W(c(u),c_i)$\\
\If{$I_{max}\leq
\frac{|\{v\in N(v) \mid c(v)\neq\varnothing\}|}{|N(v)|}t$}{
\lnl{InRes2}\ForEach{$u\in N(v) \mid c(u)\neq\varnothing$}{
   $I_{max}=\sum_{w\in N(u);c(w)\neq\varnothing}{W(c(w),c_i)}$\\
   \If{$I_{max}>
   \frac{|\{v\in N(u) \mid c(v)\neq\varnothing\}|}{|N(u)|}t$}{
   \textbf{break}
   }
}
}
}
\If{$i>|S|$}{
   $c(v)\coloneqq\varnothing,\ \forall v \in |V|$;\\
   \textbf{break}}
}

\end{algorithm2e}

For comparison with our heuristics, we have tested a generic nonlinear optimizer based in Particle Swarm Optimization (PSO)~\cite{Zhang15}.
For the TSC problem, we have used the \textit{sum} of the interferences per vertex as the objective function, since our experiments proved that PSO performed significantly worse when using the maximum interference per vertex.
For the CSC problem, we have used the same PSO optimizer as for TSC in an iterative manner, progressively increasing the number of colors $k$ from one to $|V|$, until we get a solution whose maximum interference per vertex is below the threshold~$t$.

Moreover, as a baseline reference, we have evaluated a \textit{random} coloring approach, which simply selects the color for each vertex from a uniform distribution on the color set. As in the case of PSO, for the CSC problem this reference has been run in an iterative manner with an increasing number of colors $k$ until a valid solution has been found.

\subsection{Results}

Recall that the bound in Theorem~\ref{theorem:TSCbound} was tight, since for some graphs it cannot be improved. However, this does not mean that tightness is achieved for every graph. Hence, we start our experimental results by showing, in Table~\ref{tab:results_bounds_TSC}, the theoretical bound given in Theorem~\ref{theorem:TSCbound} and the gap between the best value obtained by the different techniques under study, Random, TSC-DSATUR and PSO, expressed as a percentage of the bound. (Note that the gap is computed with real values and rounded afterwards.) It can be observed that the more complex is the graph, the tighter is the bound in Theorem~\ref{theorem:TSCbound}.

\begin{table}[!htb]
\caption{Bounds for the maximum vertex interference $T_k(G,W)$ for TSC. Each row shows a specific combination of the number of vertices $n$ and the probability of connection~$p$.}
\label{tab:results_bounds_TSC}
\renewcommand{\tabcolsep}{9pt}
\renewcommand{\arraystretch}{1.1}
\centering
\footnotesize
\begin{tabular}{|cc|cc|cc|cc|}
\cline{3-8}
\multicolumn{1}{c}{} & \multicolumn{1}{c|}{} & \multicolumn{2}{c}{k=4} & \multicolumn{2}{|c}{k=6} & \multicolumn{2}{|c|}{k=11}\\
\multicolumn{1}{c}{$n$} & \multicolumn{1}{c|}{$p$} & \multicolumn{1}{c}{Bound} & \multicolumn{1}{c|}{Gap (\%)} & \multicolumn{1}{c}{Bound} & \multicolumn{1}{c|}{Gap (\%)} & \multicolumn{1}{c}{Bound} & \multicolumn{1}{c|}{Gap (\%)}\\
\hline
\multirow{5}{*}{60} & 0.1 & 6.7 & 39.3 & 5.2 & 58.1 & 3.2 & 75.7\\
 & 0.3 & 14.9 & 26.8 & 11.6 & 35.0 & 7.1 & 47.0\\
 & 0.5 & 21.0 & 15.2 & 16.4 & 22.5 & 10.0 & 35.7\\
 & 0.7 & 27.3 & 14.4 & 21.2 & 19.3 & 13.0 & 25.5\\
 & 0.9 & 32.5 & 12.0 & 25.2 & 15.5 & 15.4 & 15.6\\
\hline
\multirow{5}{*}{70} & 0.1 & 7.4 & 34.7 & 5.7 & 48.5 & 3.5 & 72.5\\
 & 0.3 & 17.1 & 23.5 & 13.3 & 33.6 & 8.1 & 44.6\\
 & 0.5 & 25.1 & 18.0 & 19.6 & 22.1 & 11.9 & 33.2\\
 & 0.7 & 32.0 & 14.4 & 24.9 & 20.6 & 15.2 & 22.6\\
 & 0.9 & 37.7 & 11.6 & 29.3 & 16.2 & 17.9 & 15.0\\
\hline
\multirow{5}{*}{80} & 0.1 & 8.3 & 31.2 & 6.5 & 46.5 & 4.0 & 69.5\\
 & 0.3 & 19.1 & 19.8 & 14.8 & 28.6 & 9.1 & 39.5\\
 & 0.5 & 28.6 & 15.5 & 22.2 & 22.9 & 13.6 & 29.0\\
 & 0.7 & 35.8 & 14.4 & 27.8 & 18.4 & 17.0 & 21.7\\
 & 0.9 & 43.4 & 12.3 & 33.7 & 14.7 & 20.6 & 15.5\\
\hline
\end{tabular}
\end{table}

\begin{table}[!htb]

\caption{Maximum vertex interference $T_k(G,W)$ and running times for TSC with $k=4$.}
\label{tab:results_TSC4}
\renewcommand{\tabcolsep}{9pt}
\renewcommand{\arraystretch}{1.1}
\centering
\footnotesize
\begin{tabular}{|cc|cc|ccc|ccc|}
\cline{3-10}
\multicolumn{1}{c}{} & \multicolumn{1}{c|}{} & \multicolumn{2}{c}{Random} & \multicolumn{3}{|c}{TSC-DSATUR} & \multicolumn{3}{|c|}{PSO} \\
\multicolumn{1}{c}{} & \multicolumn{1}{c|}{} & \multicolumn{2}{c}{$T_k(G,W)$} & \multicolumn{2}{|c}{$T_k(G,W)$} & \multicolumn{1}{c}{Time} & \multicolumn{2}{|c}{$T_k(G,W)$} & \multicolumn{1}{c|}{Time}\\
\multicolumn{1}{c}{$n$} & \multicolumn{1}{c|}{$p$} & \multicolumn{1}{c}{avg} & \multicolumn{1}{c|}{std} & \multicolumn{1}{c}{avg} & \multicolumn{1}{c}{std} & \multicolumn{1}{c|}{avg} & \multicolumn{1}{c}{avg} & \multicolumn{1}{c}{std} & \multicolumn{1}{c|}{avg} \\
\hline
\multirow{5}{*}{60}
 & 0.1 & 6.9 & 0.7 & \textbf{4.1} & 0.5 & 2.6 ms & 5.2 & 0.5 & 3.6 s\\
 & 0.3 & 14.9 & 0.7 & \textbf{10.9} & 0.8 & 3.2 ms & 12.5 & 0.5 & 5.1 s\\
 & 0.5 & 21.6 & 1.0 & \textbf{17.8} & 1.3 & 5.5 ms & 18.3 & 0.8 & 6.2 s\\
 & 0.7 & 27.7 & 0.5 & \textbf{23.4} & 0.6 & 7.6 ms & 23.9 & 0.6 & 7.6 s\\
 & 0.9 & 32.8 & 0.5 & 28.8 & 0.5 & 9.0 ms & \textbf{28.6} & 0.1 & 8.6 s\\
\hline
\multirow{5}{*}{70}
 & 0.1 & 7.7 & 0.5 & \textbf{4.8} & 0.6 & 2.9 ms& 6.0 & 0.5 & 4.5 s\\
 & 0.3 & 17.2 & 0.6 & \textbf{13.1} & 0.9 & 4.5 ms & 14.6 & 0.6 & 6.5 s\\
 & 0.5 & 25.1 & 0.7 & \textbf{20.6} & 1.4 & 8.0 ms & 21.8 & 0.6 & 8.2 s\\
 & 0.7 & 32.2 & 0.6 & \textbf{27.4} & 0.9 & 10.4 ms & 28.2 & 0.4 & 9.9 s\\
 & 0.9 & 38.3 & 0.6 & 33.5 & 0.8 & 12.1 ms & \textbf{33.3} & 0.2 & 11.4 s\\
\hline
\multirow{5}{*}{80}
 & 0.1 & 8.5 & 0.6 & \textbf{5.7} & 0.9 & 3.6 ms & 6.8 & 0.5 & 5.7 s\\
 & 0.3 & 19.2 & 0.6 & \textbf{15.3} & 0.7 & 6.7 ms & 16.7 & 0.5 & 8.3 s\\
 & 0.5 & 28.8 & 0.7 & \textbf{24.1} & 1.5 & 10.8 ms & 25.0 & 0.4 & 11.1 s\\
 & 0.7 & 36.2 & 0.6 & \textbf{30.6} & 0.7 & 13.4 ms & 31.7 & 0.7 & 13.0 s\\
 & 0.9 & 44.1 & 0.6 & \textbf{38.0} & 0.5 & 15.7 ms & 38.3 & 0.3 & 15.6 s\\
\hline
\end{tabular}

\bigskip

\caption{Maximum vertex interference $T_k(G,W)$ and running times for TSC with $k=6$.}
\label{tab:results_TSC6}
\renewcommand{\tabcolsep}{9pt}
\renewcommand{\arraystretch}{1.1}
\centering
\footnotesize
\begin{tabular}{|cc|cc|ccc|ccc|}
\cline{3-10}
\multicolumn{1}{c}{} & \multicolumn{1}{c|}{} & \multicolumn{2}{c}{Random} & \multicolumn{3}{|c}{TSC-DSATUR} & \multicolumn{3}{|c|}{PSO} \\
\multicolumn{1}{c}{} & \multicolumn{1}{c|}{} & \multicolumn{2}{c}{$T_k(G,W)$} & \multicolumn{2}{|c}{$T_k(G,W)$} & \multicolumn{1}{c}{Time} & \multicolumn{2}{|c}{$T_k(G,W)$} & \multicolumn{1}{c|}{Time}\\
\multicolumn{1}{c}{$n$} & \multicolumn{1}{c|}{$p$} & \multicolumn{1}{c}{avg} & \multicolumn{1}{c|}{std} & \multicolumn{1}{c}{avg} & \multicolumn{1}{c}{std} & \multicolumn{1}{c|}{avg} & \multicolumn{1}{c}{avg} & \multicolumn{1}{c}{std} & \multicolumn{1}{c|}{avg} \\
\hline
\multirow{5}{*}{60}
 & 0.1 & 5.6 & 0.5 & \textbf{2.2} & 0.4 & 2.7 ms & 3.5 & 0.3 & 4.2 s\\
 & 0.3 & 12.0 & 0.4 & \textbf{7.5} & 0.6 & 3.5 ms & 9.2 & 0.4 & 5.6 s\\
 & 0.5 & 17.0 & 0.7 & \textbf{12.7} & 0.6 & 5.0 ms & 13.8 & 0.7 & 6.7 s\\
 & 0.7 & 21.8 & 0.6 & \textbf{17.1} & 0.6 & 7.6 ms & 18.0 & 0.5 & 7.7 s\\
 & 0.9 & 25.8 & 0.3 & 21.5 & 0.4 & 10.1 ms & \textbf{21.3} & 0.1 & 8.7 s\\
\hline
\multirow{5}{*}{70}
 & 0.1 & 6.3 & 0.3 & \textbf{3.0} & 0.4 & 3.1 ms & 4.4 & 0.4 & 5.1 s\\
 & 0.3 & 13.6 & 0.5 & \textbf{8.8} & 0.9 & 4.6 ms & 10.9 & 0.5 & 6.9 s\\
 & 0.5 & 20.0 & 0.5 & \textbf{15.2} & 0.7 & 7.3 ms & 16.4 & 0.3 & 8.4 s\\
 & 0.7 & 25.9 & 0.4 & \textbf{19.8} & 0.8 & 10.6 ms & 21.3 & 0.5 & 10.7 s\\
 & 0.9 & 30.1 & 0.4 & \textbf{24.6} & 0.4 & 13.6 ms & 25.1 & 0.3 & 11.5 s\\
\hline
\multirow{5}{*}{80}
 & 0.1 & 6.9 & 0.4 & \textbf{3.5} & 0.3 & 3.8 ms & 5.2 & 0.3 & 6.7 s\\
 & 0.3 & 15.3 & 0.3 & \textbf{10.6} & 1.0 & 5.8 ms & 12.5 & 0.5 & 8.4 s\\
 & 0.5 & 22.8 & 0.4 & \textbf{17.1} & 0.9 & 9.9 ms & 18.8 & 0.5 & 11.4 s\\
 & 0.7 & 28.8 & 0.7 & \textbf{22.7} & 0.8 & 13.9 ms & 24.2 & 0.5 & 12.8 s\\
 & 0.9 & 34.6 & 0.2 & \textbf{28.8} & 0.7 & 17.9 ms & \textbf{28.8} & 0.2 & 15.2 s\\
\hline
\end{tabular}
\end{table}

\begin{table}[!htb]
\caption{Maximum vertex interference $T_k(G,W)$ and running times for TSC with $k=11$.}
\label{tab:results_TSC11}
\renewcommand{\tabcolsep}{9pt}
\renewcommand{\arraystretch}{1.1}
\centering
\footnotesize
\begin{tabular}{|cc|cc|ccc|ccc|}
\cline{3-10}
\multicolumn{1}{c}{} & \multicolumn{1}{c|}{} & \multicolumn{2}{c}{Random} & \multicolumn{3}{|c}{TSC-DSATUR} & \multicolumn{3}{|c|}{PSO} \\
\multicolumn{1}{c}{} & \multicolumn{1}{c|}{} & \multicolumn{2}{c}{$T_k(G,W)$} & \multicolumn{2}{|c}{$T_k(G,W)$} & \multicolumn{1}{c}{Time} & \multicolumn{2}{|c}{$T_k(G,W)$} & \multicolumn{1}{c|}{Time}\\
\multicolumn{1}{c}{$n$} & \multicolumn{1}{c|}{$p$} & \multicolumn{1}{c}{avg} & \multicolumn{1}{c|}{std} & \multicolumn{1}{c}{avg} & \multicolumn{1}{c}{std} & \multicolumn{1}{c|}{avg} & \multicolumn{1}{c}{avg} & \multicolumn{1}{c}{std} & \multicolumn{1}{c|}{avg} \\
\hline
\multirow{5}{*}{60}
 & 0.1 & 4.0 & 0.3 & \textbf{0.8} & 0.1 & 3.2 ms & 1.9 & 0.2 & 5.5 s\\
 & 0.3 & 8.2 & 0.3 & \textbf{3.7} & 0.5 & 4.7 ms & 5.5 & 0.2 & 6.3 s\\
 & 0.5 & 11.6 & 0.4 & \textbf{6.4} & 0.4 & 6.5 ms & 8.4 & 0.4 & 7.6 s\\
 & 0.7 & 14.5 & 0.2 & \textbf{9.7} & 0.5 & 8.8 ms & 11.1 & 0.2 & 8.9 s\\
 & 0.9 & 17.0 & 0.3 & \textbf{13.0} & 0.3 & 13.0 ms & 13.3 & 0.1 & 10.3 s\\
\hline
\multirow{5}{*}{70}
 & 0.1 & 4.6 & 0.2 & \textbf{1.0} & 0.1 & 3.9 ms & 2.5 & 0.1 & 6.8 s\\
 & 0.3 & 9.3 & 0.3 & \textbf{4.5} & 0.3 & 6.1 ms & 6.6 & 0.3 & 8.3 s\\
 & 0.5 & 13.4 & 0.2 & \textbf{8.0} & 0.5 & 8.8 ms & 10.1 & 0.1 & 10.3 s\\
 & 0.7 & 17.0 & 0.4 & \textbf{11.8} & 0.7 & 12.2 ms & 13.2 & 0.2 & 11.9 s\\
 & 0.9 & 19.5 & 0.3 & \textbf{15.2} & 0.3 & 17.6 ms & 15.4 & 0.1 & 13.8 s\\
\hline
\multirow{5}{*}{80}
 & 0.1 & 4.9 & 0.2 & \textbf{1.2} & 0.2 & 4.9 ms & 3.0 & 0.2 & 7.9 s\\
 & 0.3 & 10.4 & 0.3 & \textbf{5.5} & 0.6 & 7.8 ms & 7.8 & 0.3 & 10.4 s\\
 & 0.5 & 15.1 & 0.4 & \textbf{9.6} & 0.7 & 11.4 ms & 11.7 & 0.2 & 12.8 s\\
 & 0.7 & 18.9 & 0.3 & \textbf{13.3} & 0.7 & 15.8 ms & 15.0 & 0.3 & 15.2 s\\
 & 0.9 & 22.5 & 0.4 & \textbf{17.4} & 0.7 & 23.4 ms & 17.7 & 0.1 & 17.3 s\\
\hline
\end{tabular}
\end{table}

On the other hand, Tables~\ref{tab:results_TSC4}, \ref{tab:results_TSC6}, and \ref{tab:results_TSC11} show the results obtained for the \textsc{Threshold Spectrum Coloring} with $k$ equal to 4, 6, and 11 colors, respectively, and for each graph category. In columns, we have represented the average (showing in boldface the best value of each row) and standard deviation for the maximum vertex interference $T_k(G,W)$ achieved using each evaluated strategy: the {random} reference, our implementation of {TSC-DSATUR}, and the nonlinear optimizer {PSO}. In addition, we also show the average time required to compute that interference $T_k(G,W)$.

It is important to observe that the results in Table~\ref{tab:results_TSC11} provide, for the Wi-Fi channel assignment problem in the corresponding graph, an interference threshold which can be ensured to be achievable at every vertex.

As for the algorithms, the results show that the performance in terms of $T_k(G,W)$ of our heuristic TSC-DSATUR is better than the performance of PSO, except for a very limited number of cases. Moreover, in those few cases where PSO is better than TSC-DSATUR, the difference is very small. If we analyze the running times of both techniques, we can observe that TSC-DSATUR runs about 1000 times faster than PSO. From this results we can conclude that the use of the heuristic algorithm TSC-DSATUR is very advisable, as it is able to obtain results that are usually better than those obtained by the nonlinear optimizer PSO but with much lower computation requirements.

If we compare both algorithms in more detail, and in general, we can conclude that the performance of TSC-DSATUR in comparison to PSO decreases as the complexity of the graph increases, i.e., as $p$ increases, being fairly similar both performances when $p=0.9$. This is due to the fact that PSO deals better with complex nonlinear problems. Furthermore, the standard deviations are noticeable lower in PSO than in TSC-DSATUR. As expected, the values obtained for the random algorithm are always much worse than the other approaches. In particular, it is interesting to see that for $k=4$ and $k=6$ the random algorithm achieves results very close to the theoretical bound, while with the highest number of colors that approach is worse than the bound.

It is very interesting to inspect the results also in terms of the expected average degree of the graphs, given by~$np$, as it is shown in Figure~\ref{fig:np_TSC}. In Figures~\ref{fig:np_4}-\ref{fig:np_11} we can identify a linear trend between the achieved maximum interference $T_k(G,W)$ and the expected average degree of the graphs for $k$ equal to 4, 6 and 11 colors, respectively. Finally, in Figure~\ref{fig:np_timesTSC} we show how the quotient between the running times between PSO and TSC-DSATUR decreases as $np$ increases.

\begin{figure}[!htb]
\begin{minipage}{.5\linewidth}
\centering
\subfloat[$k=4$.]{\label{fig:np_4}\includegraphics[width=1\textwidth] {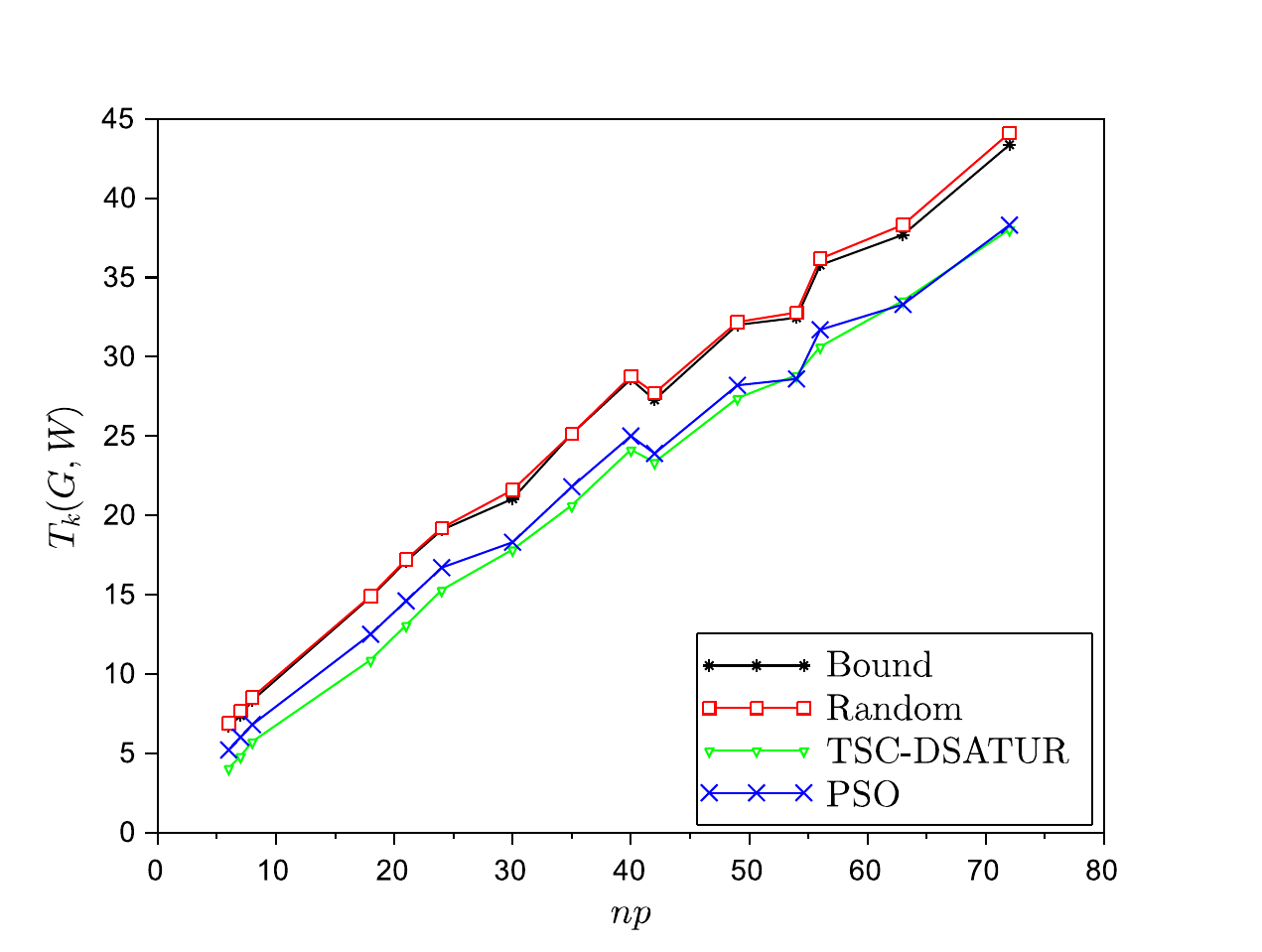}}
\end{minipage}%
\begin{minipage}{.5\linewidth}
\centering
\subfloat[$k=6$.]{\label{fig:np_6}\includegraphics[width=1\textwidth]{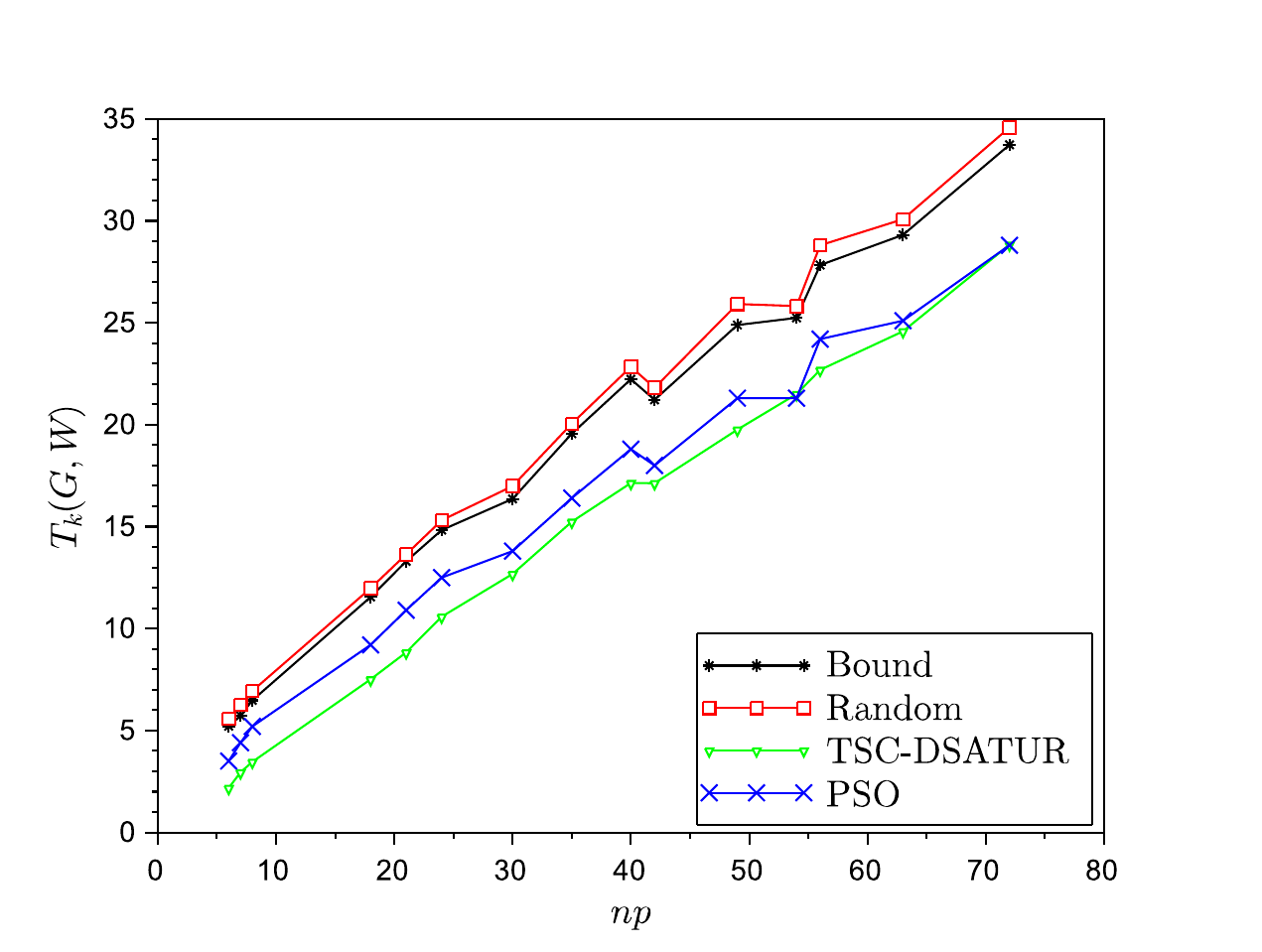}}
\end{minipage}\par\medskip
\begin{minipage}{.5\linewidth}
\centering
\subfloat[$k=11$.]{\label{fig:np_11}\includegraphics[width=1\textwidth] {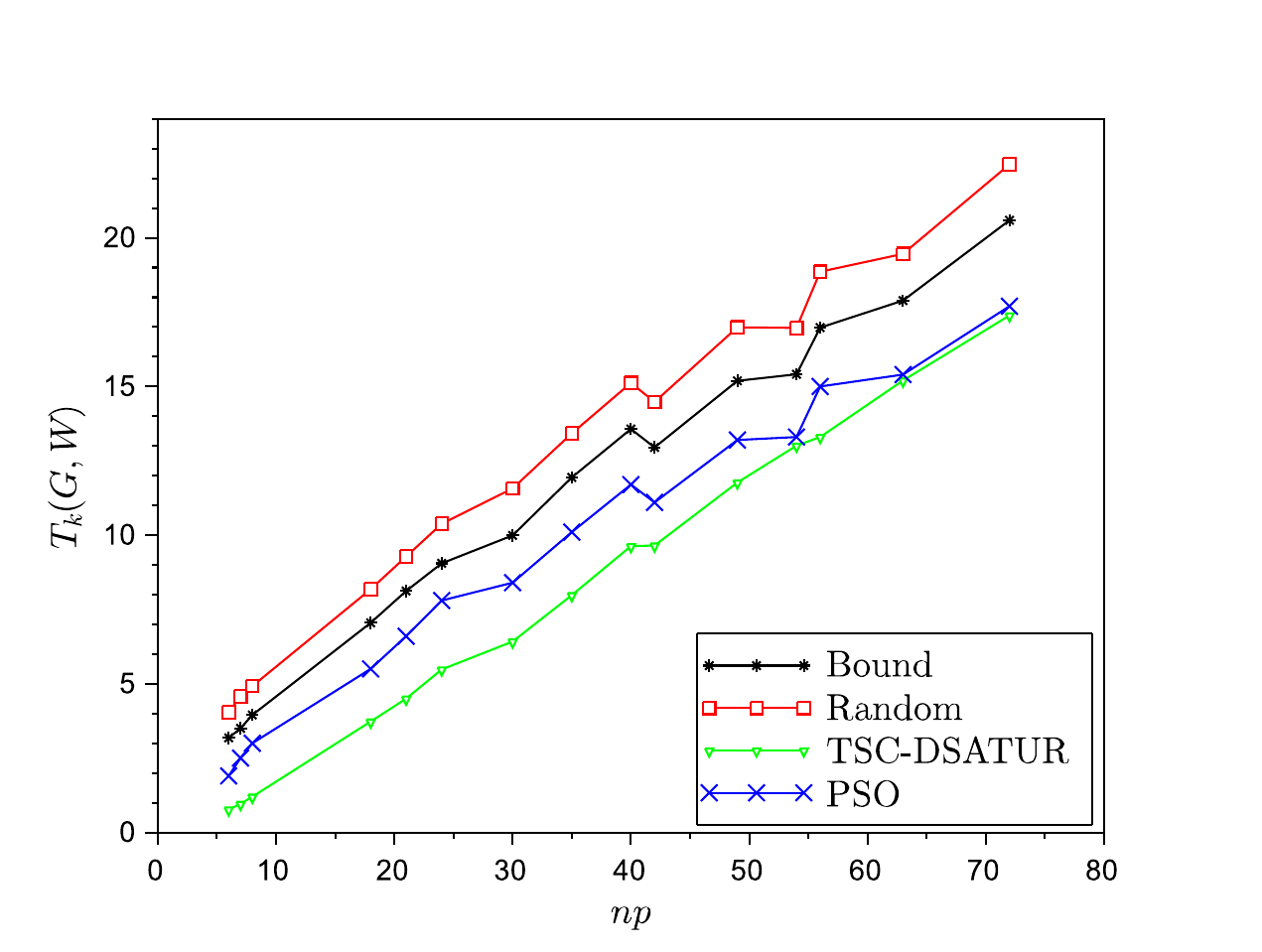}}
\end{minipage}%
\begin{minipage}{.5\linewidth}
\centering
\subfloat[Running times.]{\label{fig:np_timesTSC}\includegraphics[width=1\textwidth]{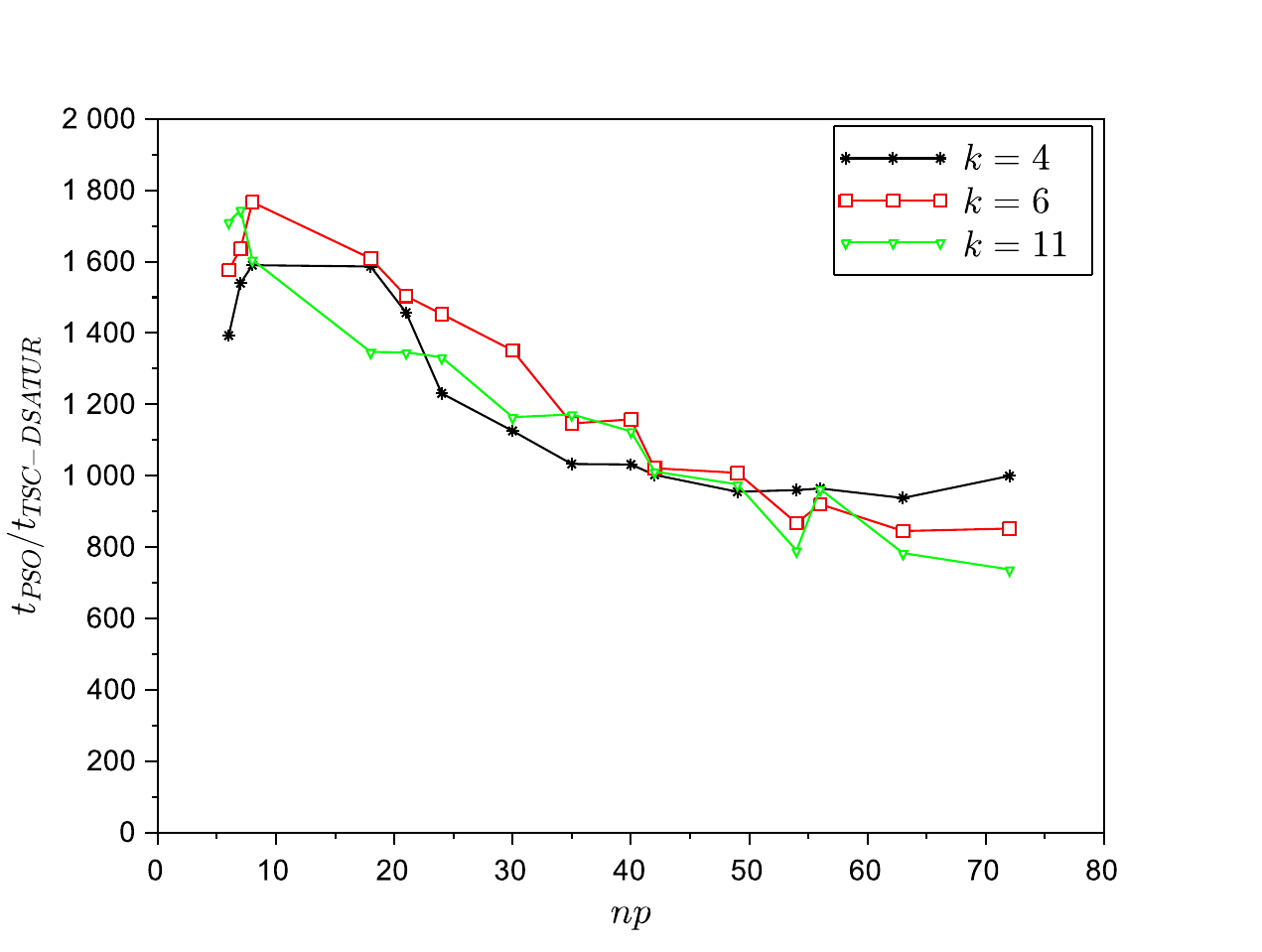}}
\end{minipage}\par\medskip
\caption{Effect of expected average degree of the graphs ($np$) in the \textsc{Threshold Spectrum Coloring} problem.}
\label{fig:np_TSC}
\end{figure}

Regarding \textsc{Chromatic Spectrum Coloring}, Table~\ref{tab:results_bounds_CSC} shows the value of the tight theoretical bound given in Theorem~\ref{theorem:CSCbound} and the gap, for the graphs considered, between the best experimental result and the theoretical bound. The main conclusion is that the gap for the theoretical upper bound again gets narrower when the complexity of the graph increases.

\begin{table}[!htb]

\caption{Bounds for the number of colors $\chi_t(G,W)$ for CSC.}
\label{tab:results_bounds_CSC}
\renewcommand{\tabcolsep}{9pt}
\renewcommand{\arraystretch}{1.1}
\centering
\footnotesize
\begin{tabular}{|cc|cc|cc|cc|}
\cline{3-8}
\multicolumn{1}{c}{} & \multicolumn{1}{c|}{} & \multicolumn{2}{c}{$t=np/4$} & \multicolumn{2}{|c}{$t=np/2$} & \multicolumn{2}{|c|}{$t=3np/4$}\\
\multicolumn{1}{c}{$n$} & \multicolumn{1}{c|}{$p$} & \multicolumn{1}{c}{Bound} & \multicolumn{1}{c|}{Gap (\%)} & \multicolumn{1}{c}{Bound} & \multicolumn{1}{c|}{Gap (\%)} & \multicolumn{1}{c}{Bound} & \multicolumn{1}{c|}{Gap (\%)}\\
\hline
\multirow{5}{*}{60}
 & 0.1 & 24 & 62.9 & 12 & 52.5 & 8 & 47.5\\
 & 0.3 & 18 & 41.1 & 9 & 31.1 & 6 & 33.8\\
 & 0.5 & 15 & 26.7 & 8 & 27.5 & 5 & 34.5\\
 & 0.7 & 14 & 17.1 & 7 & 25.3 & 5 & 40.0\\
 & 0.9 & 13 & 14.0 & 7 & 28.6 & 5 & 40.0\\
\hline
\multirow{5}{*}{70}
 & 0.1 & 23 & 58.3 & 12 & 49.2 & 8 & 48.8\\
 & 0.3 & 18 & 40.0 & 9 & 33.3 & 6 & 33.3\\
 & 0.5 & 16 & 28.8 & 8 & 25.1 & 6 & 43.3\\
 & 0.7 & 14 & 15.7 & 7 & 23.7 & 5 & 40.0\\
 & 0.9 & 13 & 15.4 & 7 & 28.6 & 5 & 40.0\\
\hline
\multirow{5}{*}{80}
 & 0.1 & 23 & 55.7 & 12 & 45.8 & 8 & 47.5\\
 & 0.3 & 17 & 35.3 & 9 & 33.3 & 6 & 33.3\\
 & 0.5 & 16 & 27.5 & 8 & 23.9 & 6 & 44.3\\
 & 0.7 & 14 & 17.9 & 7 & 24.9 & 5 & 40.0\\
 & 0.9 & 13 & 13.2 & 7 & 28.6 & 5 & 40.0\\
\hline
\end{tabular}

\bigskip

\caption{Number of colors $\chi_t(G,W)$ and running times for CSC with $t=np/4$.}
\label{tab:results_CSC025}
\renewcommand{\tabcolsep}{9pt}
\renewcommand{\arraystretch}{1.1}
\centering
\footnotesize
\begin{tabular}{|cc|cc|ccc|ccc|}
\cline{3-10}
\multicolumn{1}{c}{} & \multicolumn{1}{c|}{} & \multicolumn{2}{c}{Random} & \multicolumn{3}{|c}{CSC-DSATUR} & \multicolumn{3}{|c|}{PSO} \\
\multicolumn{1}{c}{} & \multicolumn{1}{c|}{} & \multicolumn{2}{c}{$\chi_t(G,W)$} & \multicolumn{2}{|c}{$\chi_t(G,W)$} & \multicolumn{1}{c}{Time} & \multicolumn{2}{|c}{$\chi_t(G,W)$} & \multicolumn{1}{c|}{Time}\\
\multicolumn{1}{c}{$n$} & \multicolumn{1}{c|}{$p$} & \multicolumn{1}{c}{avg} & \multicolumn{1}{c|}{std} & \multicolumn{1}{c}{avg} & \multicolumn{1}{c}{std} & \multicolumn{1}{c|}{avg} & \multicolumn{1}{c}{avg} & \multicolumn{1}{c}{std} & \multicolumn{1}{c|}{avg} \\
\hline
\multirow{5}{*}{60}
& 0.1 & 32.0 & 1.4 & \textbf{8.9} & 0.7 & 3.9 ms & 12.5 & 0.7 & 6.0 s\\
& 0.3 & 21.7 & 0.7 & \textbf{10.6} & 0.8 & 12.3 ms & 13.3 & 0.7 & 6.9 s\\
& 0.5 & 18.1 & 0.7 & \textbf{11.0} & 0.8 & 30.0 ms & 12.6 & 0.6 & 8.3 s\\
& 0.7 & 16.2 & 0.3 & \textbf{11.6} & 0.5 & 56.0 ms & 12.1 & 0.2 & 9.5 s\\
& 0.9 & 14.7 & 0.3 & \textbf{11.2} & 0.4 & 99.9 ms & \textbf{11.2} & 0.1 & 10.6 s\\
\hline
\multirow{5}{*}{70}
& 0.1 & 32.7 & 1.3 & \textbf{9.6} & 0.7 & 5.4 ms & 13.9 & 0.8 & 8.2 s\\
& 0.3 & 22.0 & 0.5 & \textbf{10.8} & 0.6 & 18.4 ms & 13.9 & 0.4 & 8.9 s\\
& 0.5 & 18.3 & 0.3 & \textbf{11.4} & 0.5 & 45.6 ms & 13.2 & 0.4 & 11.0 s\\
& 0.7 & 16.3 & 0.2 & \textbf{11.8} & 0.4 & 89.9 ms & 12.4 & 0.3 & 12.9 s\\
& 0.9 & 14.6 & 0.3 & \textbf{11.0} & 0.0 & 154.9 ms & 11.3 & 0.1 & 14.3 s\\
\hline
\multirow{5}{*}{80}
& 0.1 & 34.9 & 1.7 & \textbf{10.2} & 1.0 & 6.7 ms & 15.2 & 0.7 & 10.1 s\\
& 0.3 & 22.0 & 0.5 & \textbf{11.0} & 0.6 & 27.7 ms & 14.4 & 0.4 & 11.2 s\\
& 0.5 & 18.7 & 0.2 & \textbf{11.6} & 0.5 & 70.0 ms & 13.5 & 0.3 & 13.7 s\\
& 0.7 & 16.2 & 0.5 & \textbf{11.5} & 0.5 & 129.7 ms & 12.4 & 0.2 & 15.8 s\\
& 0.9 & 14.4 & 0.3 & 11.5 & 0.5 & 230.7 ms & \textbf{11.3} & 0.1 & 18.2 s\\
\hline
\end{tabular}
\end{table}

Tables~\ref{tab:results_CSC025}, \ref{tab:results_CSC05} and~\ref{tab:results_CSC075} show the results obtained with a threshold~$t$ on the maximum interference per vertex equal to $np/4$, $np/2$ and $3np/4$, respectively. Again, each row corresponds to a different graph category, given by $np$ and again we have marked the best values in boldface. In columns, we have represented the average and standard deviation for the number of colors $\chi_t(G,W)$ achieved for each graph coloring approach evaluated together with the running time needed to obtain those results. Keep in mind that, as mentioned in Section~\ref{sec:techniques}, both {PSO} and the {random} reference are implemented for CSC in an iterative manner, that is, they are run with an increasing number of candidate colors until we find the smaller value~$\chi_t$ for which the result satisfies the threshold $t$. We can observe again that our heuristic CSC-DSATUR generally outperforms PSO, especially for the lowest values of $p$. It is also interesting to note that the difference between the results obtained by our heuristic and the random reference are significantly lower as the average degree of the graph (which depends on~$p$) increases. Although this may seem counter-intuitive, we must keep in mind that we are adjusting the threshold~$t$ proportionally with the expected average degree~$np$, which keeps around the same values the minimum number of different colors needed by the solutions whose interference falls below the threshold. Finding such a solution randomly is a different issue, however, and it becomes easier when the expected average degree increases, since the proportion of such solutions in the global solution space increases when increasing the degrees of freedom (i.e., the number of edges) per vertex. Finally, in Figure~\ref{fig:time_CSC} we show the quotient between the running times required to obtain the results using PSO and CSC-DSATUR, concluding that the time gain obtained when using CSC-DSATUR decreases as the expected average degree of the graph, $np$, increases, although this gain is at least of two orders of magnitude.

\begin{table}[!htb]

\caption{Number of colors $\chi_t(G,W)$ and running times for CSC with $t=np/2$.}
\label{tab:results_CSC05}
\renewcommand{\tabcolsep}{9pt}
\renewcommand{\arraystretch}{1.1}
\centering
\footnotesize
\begin{tabular}{|cc|cc|ccc|ccc|}
\cline{3-10}
\multicolumn{1}{c}{} & \multicolumn{1}{c|}{} & \multicolumn{2}{c}{Random} & \multicolumn{3}{|c}{CSC-DSATUR} & \multicolumn{3}{|c|}{PSO} \\
\multicolumn{1}{c}{} & \multicolumn{1}{c|}{} & \multicolumn{2}{c}{$\chi_t(G,W)$} & \multicolumn{2}{|c}{$\chi_t(G,W)$} & \multicolumn{1}{c}{Time} & \multicolumn{2}{|c}{$\chi_t(G,W)$} & \multicolumn{1}{c|}{Time}\\
\multicolumn{1}{c}{$n$} & \multicolumn{1}{c|}{$p$} & \multicolumn{1}{c}{avg} & \multicolumn{1}{c|}{std} & \multicolumn{1}{c}{avg} & \multicolumn{1}{c}{std} & \multicolumn{1}{c|}{avg} & \multicolumn{1}{c}{avg} & \multicolumn{1}{c}{std} & \multicolumn{1}{c|}{avg} \\
\hline
\multirow{5}{*}{60}
& 0.1 & 15.0 & 1.0 & \textbf{5.7} & 0.5 & 3.4 ms & 7.4 & 0.5 & 4.9 s\\
& 0.3 & 9.4 & 0.5 & \textbf{6.2} & 0.4 & 10.1 ms & 6.6 & 0.3 & 5.8 s\\
& 0.5 & 7.6 & 0.4 & \textbf{5.8} & 0.6 & 23.7 ms & \textbf{5.8} & 0.3 & 6.8 s\\
& 0.7 & 6.8 & 0.2 & 5.4 & 0.5 & 43.3 ms & \textbf{5.2} & 0.2 & 7.7 s\\
& 0.9 & 6.1 & 0.2 & \textbf{5.0} & 0.0 & 63.8 ms & \textbf{5.0} & 0.0 & 8.5 s\\
\hline
\multirow{5}{*}{70}
& 0.1 & 14.7 & 1.1 & \textbf{6.1} & 0.3 & 4.6 ms & 8.0 & 0.6 & 6.0 s\\
& 0.3 & 9.1 & 0.4 & \textbf{6.0} & 0.6 & 15.4 ms & 6.7 & 0.3 & 7.2 s\\
& 0.5 & 7.7 & 0.2 & 6.1 & 0.3 & 34.6 ms & \textbf{6.0} & 0.2 & 8.9 s\\
& 0.7 & 6.8 & 0.3 & 5.4 & 0.5 & 65.6 ms & \textbf{5.3} & 0.2 & 10.3 s\\
& 0.9 & 6.0 & 0.1 & \textbf{5.0} & 0.0 & 99.1 ms & \textbf{5.0} & 0.0 & 11.6 s\\
\hline
\multirow{5}{*}{80}
& 0.1 & 14.0 & 0.8 & \textbf{6.5} & 0.5 & 6.0 ms & 8.1 & 0.6 & 7.0 s\\
& 0.3 & 9.0 & 0.3 & \textbf{6.0} & 0.4 & 22.4 ms & 6.7 & 0.2 & 8.7 s\\
& 0.5 & 7.6 & 0.3 & \textbf{6.1} & 0.3 & 49.9 ms & \textbf{6.1} & 0.2 & 11.1 s\\
& 0.7 & 6.7 & 0.2 & 5.7 & 0.5 & 95.6 ms & \textbf{5.3} & 0.1 & 12.6 s\\
& 0.9 & 6.2 & 0.1 & \textbf{5.0} & 0.0 & 145.4 ms & \textbf{5.0} & 0.0 & 15.1 s\\
\hline
\end{tabular}

\bigskip

\caption{Number of colors $\chi_t(G,W)$ and running times for CSC with $t=3np/4$.}
\label{tab:results_CSC075}
\renewcommand{\tabcolsep}{9pt}
\renewcommand{\arraystretch}{1.1}
\centering
\footnotesize
\begin{tabular}{|cc|cc|ccc|ccc|}
\cline{3-10}
\multicolumn{1}{c}{} & \multicolumn{1}{c|}{} & \multicolumn{2}{c}{Random} & \multicolumn{3}{|c}{CSC-DSATUR} & \multicolumn{3}{|c|}{PSO} \\
\multicolumn{1}{c}{} & \multicolumn{1}{c|}{} & \multicolumn{2}{c}{$\chi_t(G,W)$} & \multicolumn{2}{|c}{$\chi_t(G,W)$} & \multicolumn{1}{c}{Time} & \multicolumn{2}{|c}{$\chi_t(G,W)$} & \multicolumn{1}{c|}{Time}\\
\multicolumn{1}{c}{$n$} & \multicolumn{1}{c|}{$p$} & \multicolumn{1}{c}{avg} & \multicolumn{1}{c|}{std} & \multicolumn{1}{c}{avg} & \multicolumn{1}{c}{std} & \multicolumn{1}{c|}{avg} & \multicolumn{1}{c}{avg} & \multicolumn{1}{c}{std} & \multicolumn{1}{c|}{avg} \\
\hline
\multirow{5}{*}{60}
& 0.1 & 8.4 & 0.9 & \textbf{4.2} & 0.6 & 3.5 ms & 5.1 & 0.5 & 4.6 s\\
& 0.3 & 5.1 & 0.4 & \textbf{4.0} & 0.0 & 10.0 ms & \textbf{4.0} & 0.3 & 5.5 s\\
& 0.5 & 4.1 & 0.3 & 3.5 & 0.5 & 21.1 ms & \textbf{3.3} & 0.3 & 6.5 s\\
& 0.7 & 3.6 & 0.1 & \textbf{3.0} & 0.0 & 35.6 ms & \textbf{3.0} & 0.0 & 7.3 s\\
& 0.9 & \textbf{3.0} & 0.0 & \textbf{3.0} & 0.0 & 49.6 ms & \textbf{3.0} & 0.0 & 8.6 s\\
\hline
\multirow{5}{*}{70}
& 0.1 & 8.2 & 0.8 & \textbf{4.1} & 0.5 & 4.5 ms & 5.2 & 0.4 & 5.9 s\\
& 0.3 & 5.1 & 0.3 & \textbf{4.0} & 0.0 & 14.4 ms & \textbf{4.0} & 0.2 & 7.2 s\\
& 0.5 & 4.2 & 0.2 & \textbf{3.4} & 0.5 & 32.6 ms & 3.5 & 0.3 & 8.4 s\\
& 0.7 & 3.5 & 0.2 & \textbf{3.0} & 0.0 & 52.7 ms & \textbf{3.0} & 0.0 & 9.8 s\\
& 0.9 & 3.1 & 0.1 & \textbf{3.0} & 0.0 & 74.1 ms & \textbf{3.0} & 0.0 & 11.1 s\\
\hline
\multirow{5}{*}{80}
& 0.1 & 7.8 & 0.7 & \textbf{4.2} & 0.4 & 5.7 ms & 5.2 & 0.4 & 6.7 s\\
& 0.3 & 5.0 & 0.3 & \textbf{4.0} & 0.0 & 20.0 ms & 4.1 & 0.2 & 8.9 s\\
& 0.5 & 4.2 & 0.2 & 3.8 & 0.4 & 45.3 ms & \textbf{3.3} & 0.2 & 10.7 s\\
& 0.7 & 3.4 & 0.2 & \textbf{3.0} & 0.0 & 76.9 ms & \textbf{3.0} & 0.0 & 12.1 s\\
& 0.9 & 3.1 & 0.1 & \textbf{3.0} & 0.0 & 113.4 ms & \textbf{3.0} & 0.0 & 14.5 s\\
\hline
\end{tabular}
\end{table}

\begin{figure}[!htb]
\begin{center}
\includegraphics[width=0.6\textwidth]{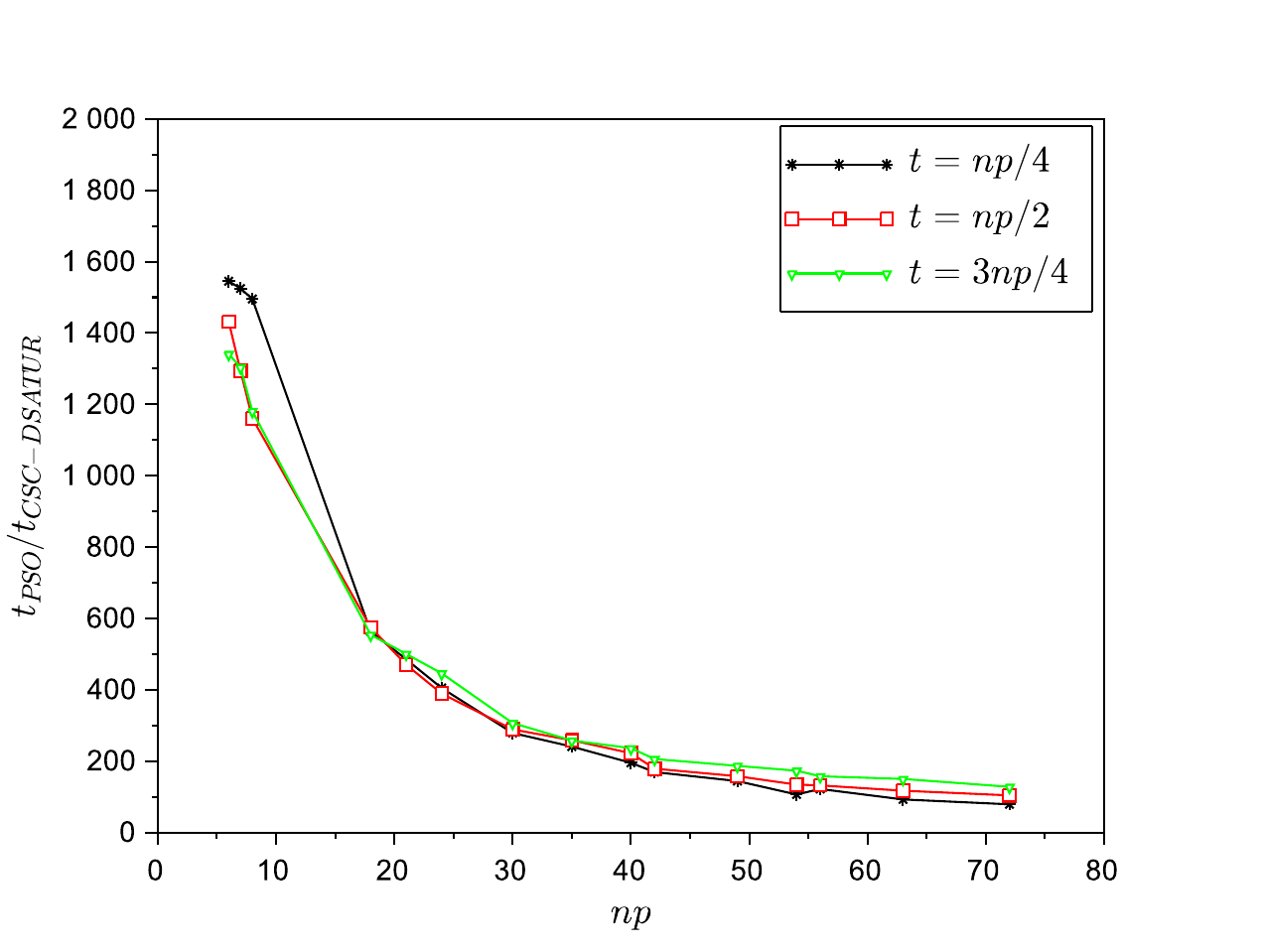}
\end{center}
\caption{Comparison between the running times of CSC-DSATUR and PSO for the CSC problem.} \label{fig:time_CSC}
\end{figure}

\section{Experimental results in a real setting}

In addition to the experimental results presented in Section~\ref{results}, we have also performed an evaluation of our proposal in a realistic setting. The problem consists of assigning a Wi-Fi channel to each of the deployed APs of a Wi-Fi network in order to minimize the interference threshold. As we have a fixed number of channels (colors) $k=11$ and want to minimize the maximum of the interferences at the vertices, we have a TSC problem.
More specifically, we have made use of the real layout of the Wi-Fi network in the Polytechnic School of our university, as shown in Figure~\ref{fig:eps}. The sides of this building, approximately square-shaped, are 130 meters long. Its main features are its 48 classrooms and its central courtyard. Moreover, we have made use of the real positions of the 26 access points (APs) deployed in the building (green dots in Figure~\ref{fig:eps}).

\begin{figure}[!htb]
\begin{minipage}{.5\linewidth}
\centering
\subfloat[Scenario plan.]{\label{fig:eps}\includegraphics[width=0.98\textwidth] {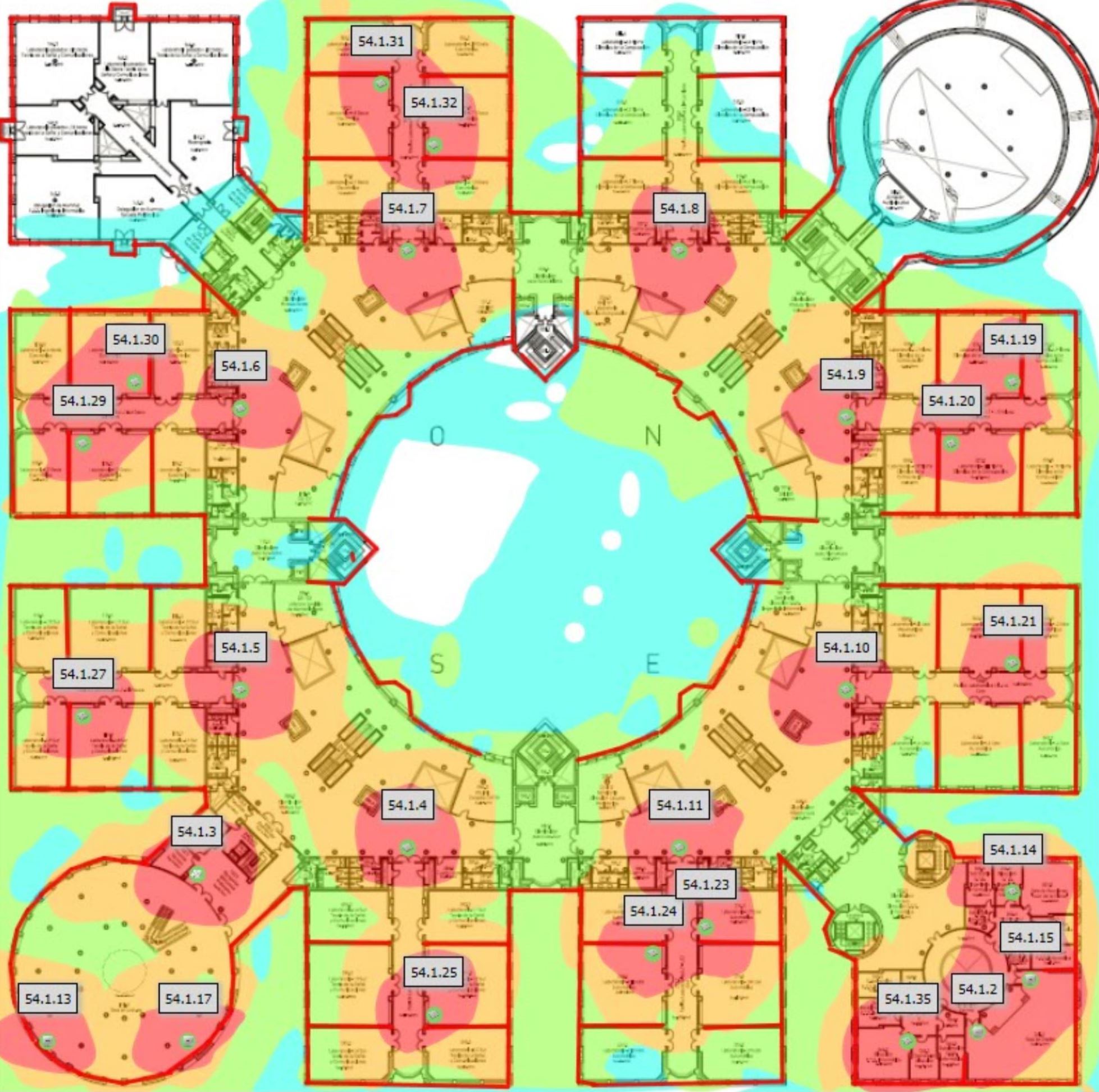}}
\end{minipage}%
\begin{minipage}{.5\linewidth}
\centering
\subfloat[Graph model for LD.]{\label{fig:eps_ld}\includegraphics[width=1\textwidth]{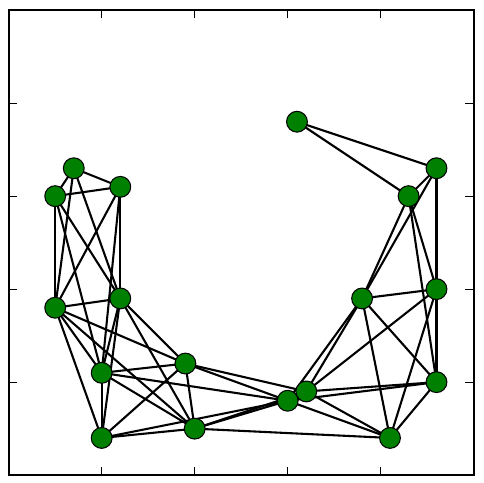}}
\end{minipage}\par\medskip
\begin{minipage}{.5\linewidth}
\centering
\subfloat[Graph model for MD.]{\label{fig:eps_md}\includegraphics[width=1\textwidth] {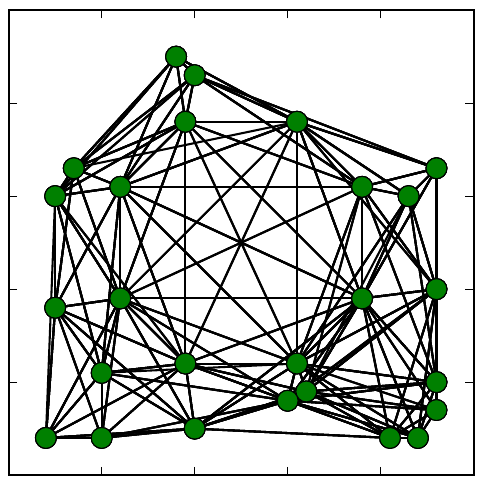}}
\end{minipage}%
\begin{minipage}{.5\linewidth}
\centering
\subfloat[Graph model for HD.]{\label{fig:eps_hd}\includegraphics[width=1\textwidth]{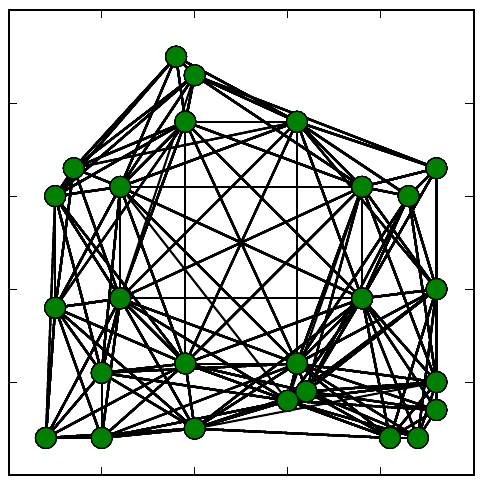}}
\end{minipage}\par\medskip
\caption{Scenario and graph models for the real setting.}
\label{fig:eps_model}
\end{figure}
The figure also shows the signal strength from each AP ranging from high coverage (red) to low coverage (blue). Note the low signal coverage for the central courtyard. For this layout, we have considered three different and representative scenarios, depending on the number of network users. In the first one, called Low Density (LD), we consider that there are only 10 classrooms in use with 20 students each, making a total of 200 users. Note that in all scenarios the position of the students in a classroom has been taken from a normal distribution centered around the center of each class and a standard deviation of 0.05 (considered that the size of the scenario is normalized to 1). In the second one, called Medium Density (MD), we consider that there are 100 students uniformly distributed in the building and also 24 classrooms in use with 25 students each, for a total of 700 network users. Finally, the third scenario, called High Density (HD), extends MD to have all the 48 classrooms in use, for a total of 1300 users. To model each scenario, we have considered a graph, whose vertices are the APs, removing those APs which do not have any client attached to it. Note that we have considered that clients are attached to their closest AP. Moreover, the edges of the graph represent the possible interferences between two APs if they use a Wi-Fi channel that is close enough. We have taken into account not only the interferences due to AP transmissions, but also from clients, i.e. two APs $i$ and $j$ will be connected in the graph if the AP $i$ or any of its attached clients interferes with the AP $j$ or with any of its clients. Figures~\ref{fig:eps_ld}, \ref{fig:eps_md}, and \ref{fig:eps_hd}, show the resulting graph for each of the three scenarios. For the interferences between channels, we have used the values empirically measured in~\cite{AkyildizWiFi}, shown in Table~\ref{tab:matrix}, which are similar to those defined in~\cite{Shrivastava08}.

\begin{table}[!htb]
\caption{Interference matrix $W_{ij}$ used in the real setting \cite{AkyildizWiFi}.}
\label{tab:matrix}
\renewcommand{\tabcolsep}{9pt}
\renewcommand{\arraystretch}{1.1}
\centering
\footnotesize
\begin{tabular}{c|c|c|c|c|c|c|c|}
\cline{2-8}
$|i-j|$ & 0 & 1 & 2 & 3 & 4 & 5 & $\geq 6$\\
\cline{2-8}
$W_{ij}$ & 1 & 0.8 & 0.5 & 0.2 & 0.1 & 0.001 & 0\\
\cline{2-8}
\end{tabular}
\end{table}

For the three above-mentioned scenarios (LD, MD, and HD), Table~\ref{tab:bounds_TSCwifi} shows the computed theoretical bounds according to Theorem~\ref{theorem:TSCbound} and the gap between the best experimental result obtained and this bound.

\begin{table}[!htb]
\caption{Bounds for the maximum vertex interference $T_k(G,W)$ for TSC in the Wi-Fi scenario with $k=11$.}
\label{tab:bounds_TSCwifi}
\renewcommand{\tabcolsep}{9pt}
\renewcommand{\arraystretch}{1.1}
\centering
\footnotesize
\begin{tabular}{|c|c|c|}
\cline{2-3}
\multicolumn{1}{c}{Scenario} & \multicolumn{1}{|c|}{Bound} & \multicolumn{1}{c|}{Gap (\%)}\\
\hline
\multicolumn{1}{|c|}{LD} & 3.1 & 45.5\\
\multicolumn{1}{|c|}{MD} & 6.5 & 44.1\\
\multicolumn{1}{|c|}{HD} & 6.1 & 37.3\\
\hline
\end{tabular}
\end{table}

As in the previous examples, the gap decreases with the increase on complexity of the graphs. On the other hand, Table~\ref{tab:results_TSCwifi} shows the average and standard deviation of the maximum vertex interference $T_k(G,W)$ and the computation times for the different techniques under study.

\begin{table}[!htb]
\caption{Maximum vertex interference $T_k(G,W)$ and running times for TSC in the Wi-Fi scenario with $k=11$.}
\label{tab:results_TSCwifi}
\renewcommand{\tabcolsep}{9pt}
\renewcommand{\arraystretch}{1.1}
\centering
\footnotesize
\begin{tabular}{c|cc|ccc|ccc|}
\cline{2-9}
\multirow{3}{*}{Scenario} & \multicolumn{2}{c}{Random} & \multicolumn{3}{|c}{TSC-DSATUR} & \multicolumn{3}{|c|}{PSO}\\
& \multicolumn{2}{|c|}{$T_k(G,W)$} & \multicolumn{2}{|c}{$T_k(G,W)$} & \multicolumn{1}{c|}{Time} & \multicolumn{2}{|c}{$T_k(G,W)$} & \multicolumn{1}{c|}{Time}\\

& \multicolumn{1}{c}{avg} & \multicolumn{1}{c|}{std} & \multicolumn{1}{c}{avg} & \multicolumn{1}{c}{std} & \multicolumn{1}{c|}{avg} & \multicolumn{1}{|c}{avg} & \multicolumn{1}{c}{std} & \multicolumn{1}{c|}{std}\\
\hline
\multicolumn{1}{|c|}{LD} & 3.9 & 0.7 & 1.8 & 0.2 & 0.7 ms & \textbf{1.7} & 0.4 & 0.9 s\\
\multicolumn{1}{|c|}{MD} & 6.5 & 1.1 & \textbf{3.6} & 0.1 & 1.4 ms & 3.7 & 0.5 & 1.6 s\\
\multicolumn{1}{|c|}{HD} & 6.9 & 1.1 & \textbf{3.8} & 0.4 & 1.4 ms & 4.0 & 0.4 & 1.9 s\\
\hline
\end{tabular}
\end{table}
Results show that, clearly, Random offers always the worst performance and that the performance of TSC-DSATUR and PSO is fairly similar, although the running times of PSO are 1316.3, 1176.2 and 1355.2 times higher than the running times required by TSC-DSATUR, for LD, MD and HD, respectively. Finally, we are interested not only in the maximum interference, but also in the interference experienced by each AP, to study if the coloring obtained by TSC-DSATUR is similar to the one obtained by PSO. Figure~\ref{fig:heat_map} studies this behavior.

\begin{figure}[!htb]
\begin{minipage}{.5\linewidth}
\centering
\subfloat[LD.]{\label{fig:heat_ld}\includegraphics[trim={2.5cm 0.5cm 1.7cm 1cm},clip,width=1\textwidth] {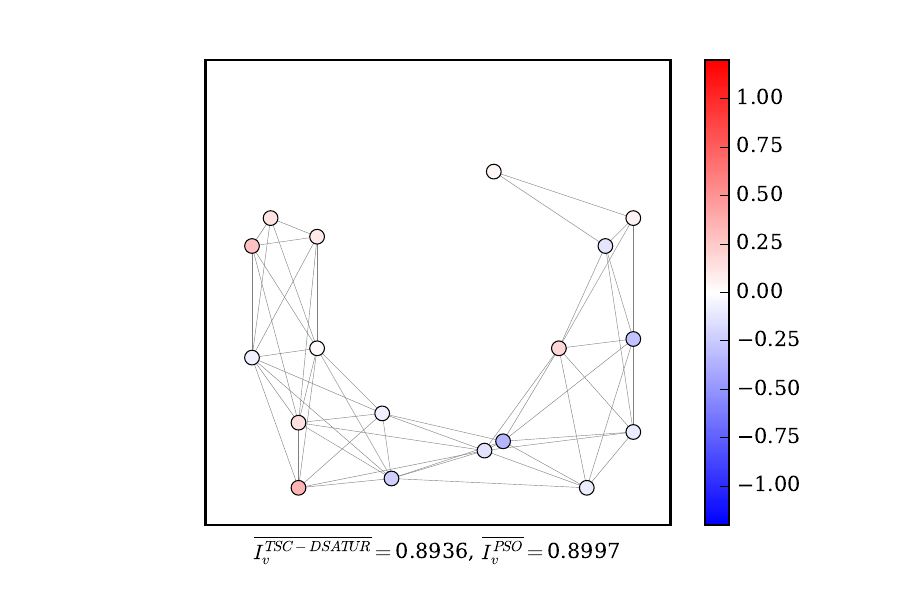}}
\end{minipage}%
\begin{minipage}{.5\linewidth}
\centering
\subfloat[MD.]{\label{fig:heat_md}\includegraphics[trim={2.5cm 0.5cm 1.7cm 1cm},clip, width=1\textwidth]{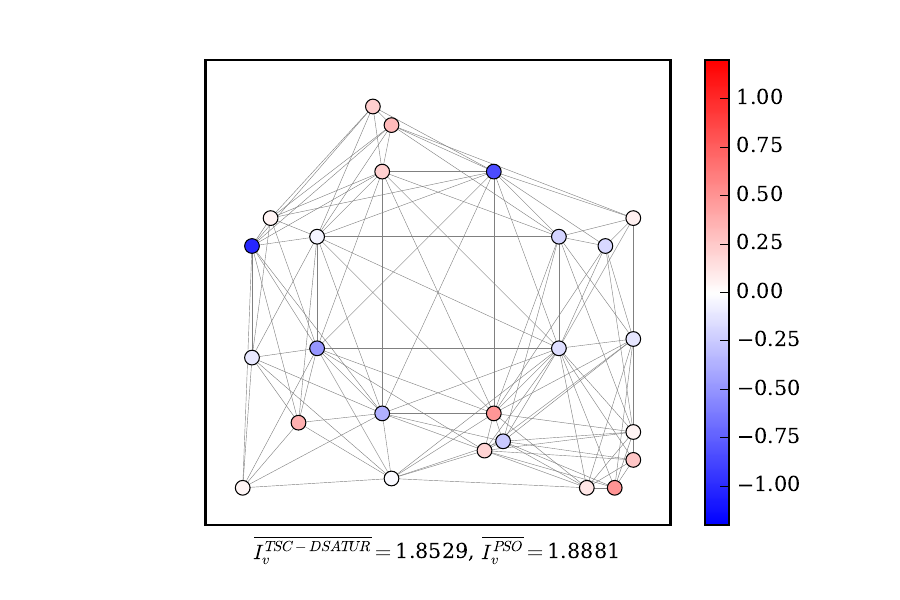}}
\end{minipage}\par\medskip
\centering
\subfloat[HD.]{\label{fig:heat_hd}\includegraphics[trim={2.5cm 0.5cm 1.7cm 1cm},clip,width=0.5\textwidth]{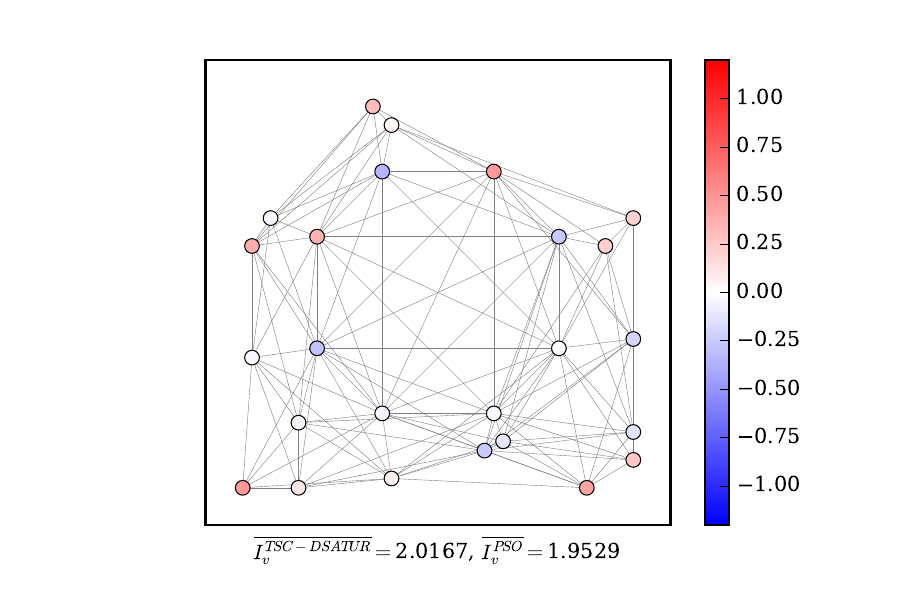}}
\caption{Heat map of the differences of interferences.}
\label{fig:heat_map}
\end{figure}

In this figure, for each AP we plot a heat map of the difference between the interference obtained by that AP using TSC-DSATUR ($I_v^{TSC-DSATUR}$) minus the interference obtained by that AP using PSO ($I_v^{PSO}$). Hence, in this figure an AP in red means that the interference of TSC-DSATUR is higher than that of PSO, and the opposite holds for blue. Results show that, in general, most of the colors are either light blue, light red or white, so the difference in performance between TSC-DSATUR and PSO is low. However, it is quite interesting to see that, for the MD scenario, there two dark blue APs, for which the performance of TSC-DSATUR is much higher than the one of PSO. Finally, in the lower part of Figure~\ref{fig:heat_map} we also show the average interference experienced by each AP. As we can see, TSC-DSATUR behaves slightly better than PSO for LD and MD, while in HD PSO has a better performance than TSC-DSATUR.

\section{Conclusions}

We have introduced two vertex-coloring problems for graphs, in which we are given a spectrum of colors endowed with a matrix of interferences between them. The \textsc{Threshold Spectrum Coloring} problem (TSC), motivated by Wi-Fi channel assignment, looks for the smallest threshold below which the interference at every vertex can be kept. The \textsc{Chromatic Spectrum Coloring} problem (CSC) is the complement, in which a threshold is given and the aim is to find the smallest number of colors allowing to respect that threshold. We have illustrated these two problems with a case study.

For such problems we have provided theoretical results, showing that both of them are NP-hard and proving upper bounds for the smallest threshold in TSC and the smallest number of colors in CSC.

In order to complete the scene, we have first presented experimental results for different coloring techniques, tested on Erd\H{o}s-Renyi random graphs of $n\in\{60,70,80\}$ vertices and probabilities $p \in \{0.1, 0.3, 0.5, 0.7, 0.9\}$  of connection, with an interference matrix of exponential decay. For the coloring techniques, we have proposed a DSATUR-based heuristic for each of the TSC and CSC problems, in order to account for the given spectrum. We have compared these heuristics with PSO, a nonlinear optimizer based in particle swarm, for TSC problems of $k\in\{4,6,11\}$ colors and CSC problems of thresholds $t\in\{np/4,np/2,3np/4\}$, where $np$ is the expected average degree of the graph.

Further, we have performed an evaluation of our proposal in a realistic setting, the Wi-Fi network in the Polytechnic School of our university.
We have made use of the real positions of the 26 access points (APs) deployed in the building and we have considered three different and representative scenarios, with 200, 700, and 1300 network users respectively.

The experiments show that the performance of our heuristics TSC-DSATUR and CSC-DSATUR is
better than PSO for the simplest graphs, although this improvement decreases as the complexity of the graphs increase. This makes sense, since higher degrees involve more interdependencies between vertex colorings, which is the kind of constraints nonlinear optimizers are designed for.

Finally, we have also checked the gap between the theoretical upper bounds and the best values obtained in the experiments, observing that the more complex is the graph, the smaller is this gap.

\section{Acknowledgements}

All the authors are supported by MINECO Projects TIN2014-61627-EXP and TIN2016-80622-P (AEI/FEDER, UE), and by the University of Alcal\'a project CCG2016/EXP-048. In addition, David Orden is supported by MINECO Project MTM2014-54207 and by the European Union's Horizon 2020 research and innovation programme under the Marie Sk\l{}odowska-Curie grant agreement No 734922.

\bibliographystyle{plain}

\end{document}